\documentclass[12pt, draftclsnofoot, onecolumn]{IEEEtran}

\usepackage{amssymb}
\usepackage{amsfonts}
\usepackage{epsf}
\usepackage{epsfig}
\usepackage{epstopdf}
\usepackage{array}
\usepackage{amsmath}
\usepackage{url}
\usepackage[algoruled, linesnumbered]{algorithm2e}
\usepackage{amsthm}
\usepackage{xifthen}
\usepackage{mathrsfs}
\usepackage[mathscr]{euscript}
\usepackage[sans]{dsfont}
\usepackage[dvipsnames]{xcolor}
\usepackage{graphicx}
\usepackage{caption}
\usepackage{subcaption}
\usepackage{footnote}

\usepackage{epsf}
\usepackage{epsfig}

\usepackage{array}

\usepackage{times}
\usepackage{graphicx}
\usepackage{amsmath}
\usepackage{amssymb}
\usepackage{amsxtra}
\usepackage{here}
\usepackage{rawfonts}
\usepackage{times}
\usepackage{url}
\usepackage{cite}
\usepackage{amsthm}
\usepackage{graphicx}
\usepackage{caption}

\usepackage{mathtools}
\usepackage{url}
\usepackage[algoruled,linesnumbered]{algorithm2e}

\usepackage{cuted}
\usepackage{multicol}

\usepackage{cite} 

\usepackage[T1]{fontenc}
\usepackage[utf8]{inputenc}
\usepackage{authblk}

\usepackage{color,soul}
\usepackage{lipsum}

\newtheorem{thm}{Theorem}
\newtheorem{lem}{Lemma}
\newtheorem{prop}{Proposition}
\newtheorem{cor}{Corollary}

\theoremstyle{definition}
\newtheorem{defn}{Definition}

\ifCLASSINFOpdf
\else
\fi

\onecolumn

\begin{document}
%
\title{Cognitive Hierarchy Theory for Distributed Resource Allocation in the Internet of Things}

\author[D. Kopta et al.]
{Nof Abuzainab$^1$, Walid Saad$^1$, Choong Seon Hong$^2$, and H. Vincent Poor $^3$
	\\ \vspace{-0.3 cm}
	\small$^1$Wireless@VT, Department of Electrical and Computer Engineering, Virginia Tech, Blacksburg, VA, USA,\\ \vspace{-0.2 cm}Emails:\{nof, walids\}@vt.edu\\ \vspace{-0.2 cm}
	$^2$Department of Computer Science and Engineering, Kyung Hee University, Yongin, South Korea, \\ \vspace{-0.2 cm} Email: cshong@khu.ac.kr\\ \vspace{-0.2 cm}
	$^3$Department of Electrical Engineering, Princeton  University, Princeton, NJ, USA, Email: poor@princeton.edu\\
	\vspace{-2.5ex}}

\author{Nof~Abuzainab,~
        Walid~Saad,~\IEEEmembership{Senior~Member,~IEEE,}
        Choong~Seong~Hong,~\IEEEmembership{Senior~Member,~IEEE,}
        ~and~H. ~Vincent~Poor,~\IEEEmembership{Life~Fellow,~IEEE} \vspace{-2.5ex}
\thanks{N. Abuzainab and W. Saad are with the department of Electrical and Computer Engineering, Virginia Tech, Blacksburg, VA, USA, e-mail: \{nof, walids\}@vt.edu}
\thanks{C.S. Hong is with department of Computer Science and Engineering, Kyung Hee University, Yongin, South Korea, email:  cshong@khu.ac.kr }
\thanks{H. V. Poor is with the department of Electrical Engineering, Princeton  University, Princeton, NJ, USA, Email: poor@princeton.edu}}
\normalsize
\maketitle


%
\IEEEpeerreviewmaketitle
\vspace{-0.5 cm}
\begin{abstract}
	In this paper, the problem of distributed resource allocation is studied for an Internet of Things (IoT) system, composed of a heterogeneous group of nodes compromising both machine-type devices (MTDs) and human-type devices (HTDs). The problem is formulated as a noncooperative game between the heterogeneous IoT devices that seek to find the optimal time allocation so as to meet their quality-of-service (QoS) requirements in terms of energy, rate and latency. Since the strategy space of each device is dependent on the actions of the other devices, the generalized Nash equilibrium (GNE) solution is first characterized, and the conditions for uniqueness of the GNE are derived. Then, to explicitly capture the heterogeneity of the devices, in terms of resource constraints and QoS needs, a novel and more realistic game-theoretic approach, based on the behavioral framework of cognitive hierarchy (CH) theory, is proposed. This approach is then shown to enable the IoT devices to reach a CH equilibrium (CHE) concept that takes into account the various levels of rationality corresponding to the heterogeneous computational capabilities and the information accessible for each one of the MTDs and HTDs.
Simulation results show that the CHE solution maintains stable performance. In particular, the proposed CHE solution keeps the percentage of devices with satisfied QoS constraints above $96\%$ for IoT networks containing up to 10,000 devices without considerably degrading the overall system performance in terms of the total utility. Simulation results also show that the proposed CHE solution brings a two-fold increase in the total rate of HTDs and deceases the total energy consumed by MTDs by $78\%$ compared to the equal time policy. 
\end{abstract}

\let\thefootnote\relax\footnote{A preliminary version of this work was presented in \cite{isit2016} at the 2016 IEEE International Symposium on Information Theory.
	
This research was supported by the U.S. Office of Naval Research (ONR) under Grant N00014-15-1-2709 and by the U.S. National Science Foundation under Grants CNS-1460333, CNS-1702808, ECCS-1647198 and OAC-1541105.}


%
\vspace{-1 cm}\section{Introduction}
\IEEEPARstart{M}{eeting} the stringent quality-of-service (QoS) requirements of a massive number of heterogeneous devices is the main challenge facing the successful deployment of the Internet of Things (IoT)\cite{isit2016,iott, niyato, m2msurvey}. In particular, the IoT ecosystem will encompass both human type devices (HTDs) and machine type devices (MTDs). MTDs are expected to deliver a wide range of applications ranging from healthcare to smart homes and transportation and, as such, they will exhibit a heterogeneous mix of quality-of-service (QoS) requirements \cite{m2msurvey}. In general, the three main important performance metrics for MTDs are reliability, energy efficiency, and latency. MTDs, such as those used for environmental monitoring, cannot be easily charged, and, thus they must minimize their energy consumption. In contrast, MTDs that are used for critical applications such as alarm systems will be primarily seeking to reliably deliver their packets under stringent delay constraints. 
HTDs, such as smartphones, typically require high-speed transmission rates while not exceeding a certain energy budget. Beyond this heterogeneous nature of the IoT, its massive scale will significantly increase the competition on the wireless resources. This, in turn, requires designing very efficient resource allocation schemes tailored to the scale and heterogeneity of the system. Moreover, in the IoT, a centralized approach to resource allocation can be prohibitive as it requires solving an optimization problem with a large number of variables and constraints corresponding to all of the IoT devices. Thus, there is a need to adopt distributed and self-organizing IoT resource allocation schemes \cite{iotson}. In addition to being distributed, resource allocation in the IoT must also explicitly factor in the various IoT devices constraints that stem from their different QoS needs and computational capabilities in terms of memory and processing powers. 
 
To address these challenges, there is a need for a joint design of resource allocation and multiple access for the IoT \cite{popovski,tai,iotn1,iotn2,iotn4,iotn3, iot4,iot2,iotn5, iot3, ram1}. In \cite{ram1}, random access was initially proposed as a suitable multiple access scheme for a system with massive number of devices such as the IoT since it does not require coordination and can properly cope with the bursty nature of the MTDs traffic. However, in a dense IoT system, random access can potentially lead to increased collisions, and thus, not all devices will be able to meet their QoS requirements \cite{ramcollisions}. For example, the performance of MTDs that require ultra low latency or HTDs that require high data rates can can be severely affected by collisions. Thus, there is a need to design a new, distributed IoT multiple access scheme that can satisfy the requirements of devices with strict QoS constraints.

\vspace{-0.4 cm}
\subsection{Related Works}
There has been significant recent interest in developing resource allocation mechanisms suitable for the IoT such as in \cite{iotn1,iotn2,iotn3,iotn4,iotn5, iot2,iot4,iot3, ram1}. Centralized scheduling schemes for  IoT LTE networks are proposed in \cite{iotn2, iotn4,iotn3,iotn1}. In \cite{iotn1}, a resource management scheme that dynamically allocates time resources between MTDs and HTDs based on current traffic conditions and QoS requirements. The works in \cite{iotn2} and \cite{iotn4} propose schemes that allocate the LTE resources to MTDs and HTDs based on a bipartite graph. In \cite{iotn3}, the authors propose two seperate uplink scheduling schemes for HTDs and MTDs in an LTE system based on channel conditions and delay requirements while taking fairness into account. Other works such as in \cite{iot4} and \cite{iot2} adopted game-theoretic approaches for distributed resource allocation problems in the IoT. The authors in \cite{iot4} study the problem of throughput maximization of MTDs under random access.  However, in \cite{iot4}, devices are considered of equal capability and similar QoS requirements. The work in \cite{iot2} considers a heterogeneous system of MTDs in which nodes can use different routing and network coding schemes to optimize heterogeneous QoS requirements. An evolutionary game approach is proposed in \cite{iotn5} for optimizing the transmission strategy in a device-to-device (D2D) enabled LTE network in which each MTD chooses to act noncooperatively or cooperatively with other MTDs using the same resources.
In \cite{iotn6}, the problem of uplink user association of IoT devices is studied in a dense small cell network using a mean-field game. 

However, these works \cite{iotn1,iotn2,iotn3,iotn4,iotn5, iot2,iot4,iot3, ram1} typically rely on the concept of a Nash equilibrium to solve the studied resource allocation problems. Such a conventional Nash equilibrium solution may not be suitable to solve the distributed resource allocation when the optimization is done at  the IoT devices' end. This is because the Nash equilibrium assumes that players have equal and similar capabilities, while IoT devices are heterogeneous and have different computational capabilities. Thus, the rationality of each IoT device is bounded by its computational capabilty.  Also, at the Nash equilibrium, in order to compute its best resource strategy, each device must know the true actions of other devices in the system, which is not practical for the IoT. Indeed, in practice, to gather information on all the actions of the opposing devices, a significant amount of information exchange (and, thus, delay) is required given the massive number of devices in the IoT. Moreover, MTDs, such as small sensors, have limited memory and are unable to store the actions of all other IoT devices. Further, mean-field game solutions such as in In \cite{iotn6} typically rely on the exchangeability property which presumes that all devices are homogeneous and have a similar impact on the game. This assumption is generally not appropriate for heterogeneous IoT environments. These limitations motivate the need for new solution approaches that are tailored to the unique nature of the IoT, in terms of heterogeneity, scale, and device constraints.

\vspace{-0.4 cm}
\subsection{Contribution}
The main contributions of this paper are summarized as follows:
\begin{itemize}

 \item We propose a distributed self-organizing resource allocation scheme that enables the IoT devices to find their optimal allocation of time resources, while explicitly catering for the individual capability of each device and the limited availability of information. In particular, we consider the problem of resource allocation in the uplink of an IoT network that is supporting a heterogeneous mix of IoT devices using a time division multiple access (TDMA) scheme. In this network, each IoT device is self interested in determining the optimal time fraction that meets its own strict QoS guarantees in a distributed manner. In particular, HTDs seek to maximize their data rate while MTDs must deliver their packets within stringent deadlines. For both HTDs and MTDs, the proposed model accounts for energy efficiency.
\item Due to the dependence of the optimal time fraction of each device on the chosen time fractions by the remaining devices,
we formulate the problem as a noncooperative game$^1$\footnote[1]{$^1$ Here, we note that, within the scope of this work, devices are assumed to be noncooperative, due to the overhead associated with cooperation in a massive IoT.} in which the IoT devices are the players that seek to guarantee their QoS while explicitly factoring in their heterogeneous requirements and devices' capabilities.
We characterize the generalized Nash equilibrium (GNE) solution of the proposed game and show the conditions under which the GNE is unique. Further, we present a learning algorithm that allows the IoT devices to reach the GNE in a distributed manner. We show that the computational complexity of the GNE learning algorithm is polynomial in the number of IoT devices. As such, given the massive scale of the IoT, it may not be feasible for IoT devices with limited computation capabilities to find the GNE.
\item To address these computational limitations, we propose a novel and more realistic solution that extends the framework of \emph{cognitive hierarchy (CH) theory} \cite{Camerer_acognitive}, a branch of behavioral game theory that assumes that players belong to different discrete levels of rationality, and that the players have different beliefs about the remaining players depending on their rationality level. The rationale behind using a CH approach is two-fold: a) The different CH rationality levels allow us to account for the heterogeneous computational capabilities of the IoT devices, while the different players' beliefs are dependent on the resources available at each IoT devices to obtain and store the information about the remaining players. Thus, the CH approach is more realistic model than the classical GNE that assumes that all players are fully rational and of equal capabilities and b) In the proposed CH approach, each device selects its equilibrium strategy based on its beliefs about the remaining devices, and, thus, it does not incur massive information exchange to find the CH equilibrium as is the case of finding the GNE solution.
These characteristics make the CH framework suitable for this problem. To model the IoT resource allocation problem using CH, we extend the conventional framework to explicitly take into account the awareness of each device of other devices at the same rationality level. This extension is important in our problem as multiple devices of the same characteristics and computational capabilities can exist in an IoT network. 

\item We show that the computational complexity of the optimization done at each device to find the cognitive hierarchy equilibrium (CHE) is much smaller than the GNE and is linear in the number of CH levels.
\item Simulation results show that the CHE solution maintains stable performance. In particular, the proposed CHE solution keeps the percentage of devices with satisfied QoS constraints above $96\%$ for IoT networks with up to 10,000 devices while not conisderably degrading the overall system performance, in terms of the total utility. Simulation results also show that the proposed CHE solution brings a two-fold increase in the total rate of HTDs and decreases the total energy by MTDs by $78\%$ compared to an equal time policy. 
\end{itemize}
  
The rest of the paper is organized as follows. Section II presents the system model and the heterogeneous time allocation (HTA) problem. Section III describes the noncooperative game formulation of the HTA problem and presents the GNE and the CHE solutions respectively. Simulation results are presented in Section IV. Finally, conclusions are drawn in Section V.
\vspace{-0.3 cm}
\section{System Model}

Consider the uplink of an IoT system composed of a heterogeneous mix of machine type devices and human type devices. In this model, we consider MTDs having strict delay requirements and HTDs having high data rate requirements that are served by a base station (BS) according to a TDMA scheme. In this scheme, transmissions occur in time periods of $T$ seconds, and each IoT device transmits during a fraction $\tau_i$ of $T$. We denote by $\mathcal{L}$ the set of $L$ IoT devices  $(\mathcal{L} = \mathcal{H} \cup \mathcal{M}$) that includes the set $\mathcal{H}$ of HTDs and the set $\mathcal{M}$ of MTDs. All IoT devices in $\mathcal{L}$ transmit on the same frequency band of bandwidth $W$. Each device $i$ transmits with a fixed power $P_i$. The channel gains, $h_i$, between any device $i \in \mathcal{L}$ and the BS are assumed to be independent block Rayleigh fading with variance $\alpha_i^2$. It is assumed that statistical channel state information (CSI) is available at each device i.e. each device knows the channel statistics but not the instantaneous channel gain. This assumption is suitable for uplink in IoT as obtaining statistical CSI requires less communication overhead than full CSI \cite{chanstate}. Additive white Gaussian noise of variance $\sigma^2$ is present at each receiver.


Each device $i \in \mathcal{L}$ must determine the optimal time fraction $\tau^*_i$ that meets its quality-of-service requirements. Due to the heterogeneity of devices, the QoS requirement of each device will depend on its type. The details of the QoS requirements  of HTDs and MTDs are presented respectively as follows.

\subsection{HTDs QoS Requirements}
HTDs, such as smartphones, will seek to maximize their expected transmission rate while not exceeding an energy budget $E_i$. The achieved rate is related to the received signal-to-noise ratio (SNR) through the ergodic capacity formula and is given by:\begin{equation}
R_i=\mathbb{E}_{h_i}\left[W\log\Big(1+\frac{|h_i|^2P_i}{\sigma^2}\Big)\right]\tau_i.
\end{equation}

The energy spent by HTD $i$ during time fraction $\tau_i$ is:
\vspace{-0.3 cm}
\begin{equation}
\mathbb{E}[\xi_i]=P_i \tau_i  T.
\vspace{-0.1 cm}
\end{equation}
Thus, each HTD $i$ solves the following optimization problem
\begin{subequations}
	\begin{align}
&\max_{\tau_i} \mathbb{E}_{h_i}\left[W\log\Big(1+\frac{|h_i|^2P_i}{\sigma^2}\Big)\right]\tau_i,\\
& \textrm{s.t.}\hspace{0.3 cm} P_i \tau_i \leq \frac{E_i}{ T}, \hspace{0.2 cm} \tau_i \leq 1-\sum_{j \in \mathcal{L}}\tau_j, \hspace{0.2 cm} 0 \leq \tau_i \leq 1, \label{HTDopt}
	\end{align}
\end{subequations}
\vspace{-0.2 cm}
where $W$ is the bandwidth. 
\vspace{-0.2 cm}

\subsection{MTDs QoS Requirements}

Each MTD $i$ seeks to deliver a packet of size $b_i$ bits within a strict delay constraint. Due to fading, the packet may not be delivered successfully within a single transmission, and, thus, the probability $p_i$ with which the packet transmitted by device $i$ is successfully decoded at the BS is defined as the probability that the received SNR is greater than a required threshold $\gamma_i$. The threshold $\gamma_i$ is chosen so that $b_i$ bits can be transmitted successfully. Here, using Shannon's capacity formula, we have\vspace{-0.2 cm}\begin{equation}
\frac{b_i}{ T \tau_i}= W\log (1+\gamma_i),\label{thres}
\end{equation} where $W$ is the bandwidth. The probability of successful decoding is therefore
\begin{eqnarray}
p_i=\textrm{Pr}\bigg[\frac{|h_i|^2P_i}{\sigma^2}\geq \gamma_i\bigg]=e^{-\frac{\gamma_i \sigma^2}{\alpha_i^2 P_i}},
\end{eqnarray}
where $P_i$ is the transmit power of MTD $i$. MTD $i$ uses retransmissions in order to reliably deliver the packet. Hence, the probability distribution of the packet success time $T_i$ of device $d_i$ follows a geometric distribution:
\begin{eqnarray}
\textrm{Pr}[T_i=k]=(1-p_i)^{k-1}p_i=\Big(1-e^{-\frac{\gamma_i \sigma^2}{\alpha_i^2 P_i}}\Big)^{k-1}e^{-\frac{\gamma_i \sigma^2}{\alpha_i^2 P_i}}.
\end{eqnarray}
where $T_i$ is the number of time slots spent to deliver the packet successfully.

Each packet of MTD $i$ should be delivered with a strict deadline of $d_i$ seconds. It has been shown in \cite{m2mtraffic} that the MTDs' traffic is bursty. Thus, the inter-arrival times of the MTD packets are assumed to be large. Hence, the queuing delay can be ignored as it tends to be much smaller than the transmission delay. 
Indeed, for characterizing latency, in such a scenario, we allow each MTD to ensure that its time to transmit does not exceed a certain threshold. Since the time to deliver each packet successfully is random, the probability that the successful transmission time exceeds a certain number of time slots $t_i$ must be very small i.e. $\textrm{Pr}[T_i \geq t_i] \leq \epsilon$ with
\begin{eqnarray}
\textrm{Pr}[T_i \geq t_i]=(1-p_i)^{t_i}
= \Big(1-e^{-\frac{\gamma_i \sigma^2}{\alpha_i^2P_i}}\Big)^{t_i}.
\end{eqnarray}
and $t_i=\frac{d_i}{T}$. For MTDs, there is a need to maintain a low energy consumption to extend the overall lifetime of the devices. Thus, the expected energy consumed by each MTD $i$ will be\vspace{-0.2 cm}\begin{eqnarray}
\mathbb{E}[\mathcal{E}_i]=P_i \times \mathbb{E}[T_i] =\frac{P_i \tau_i T}{p_i}.
\label{energyi}
\end{eqnarray}
Thus, in order to find its optimal transmission parameters, each MTD $i$ must solve the following optimization problem
\begin{subequations}
	\begin{align}
	&\min_{\tau_i} \frac{P_i \tau_i T}{p_i},\\
	&\hspace{0.2 cm} \textrm{s.t.} \hspace{0.3 cm} \textrm{Pr}[T_i \geq t_i] \leq \epsilon, \hspace{0.2 cm} \tau_i \leq 1-\sum_{j \neq i} \tau_j, \hspace{0.2 cm} 0 \leq \tau_i \leq 1. \label{MTDopt1}
	\end{align}
\end{subequations}
where the constraint in (\ref{MTDopt1}) is based on the assumption of static allocation of time fractions i.e. each device $i$ is assigned the fraction $\tau_i$ even after it succeeds in transmitting its packet.
Based on (\ref{thres}), the time fraction $\tau_i$ is a decreasing function of $\gamma_i$. Also, the expression of $\textrm{Pr}[T_i \geq t_i]$ as a function of $\tau_i$ is complicated and the probability $p_i$ is a function of $\gamma_i$. Hence for convenience, the optimization variable is changed from $\tau_i$ to $\gamma_i$ yielding
\begin{subequations}
\begin{align}
&\min_{\gamma_i}\frac{ P_i b_i}{W\log(1+\gamma_i)e^{-\frac{\gamma_i \sigma^2}{\alpha_i^2P_i}}},\\
& \textrm{s.t.} \hspace{0.1 cm} \Big(1-e^{-\frac{\gamma_i \sigma^2}{\alpha_i^2 P_i}}\Big)^{t_i} \leq \epsilon, \hspace{0.3 cm} \frac{b_i}{T} \leq W\log (1+\gamma_i), \\
&\frac{b_i}{T \cdot W\log(1+\gamma_i)}\leq 1-\sum_{j \neq i}\frac{b_j}{T \cdot W\log(1+\gamma_j)}.
\end{align}\label{MTDopt2}
\end{subequations}

By rewriting the constraints in (\ref{MTDopt2}), the optimization becomes
\begin{subequations}
\begin{align}
	&\hspace{-3 cm}\min_{\gamma_i}\frac{ P_i b_i}{W\log(1+\gamma_i)e^{-\frac{\gamma_i \sigma^2}{\alpha_i^2P_i}}},\\ \vspace{0.2 cm}
	& \hspace{-3 cm}\textrm{s.t.} \hspace{0.2 cm} \gamma_i \leq -\frac{\alpha_i^2 P_i\log(1-\sqrt[t_i]{\epsilon})}{\sigma^2},\label{gammaub}\\
	& \hspace{-2.3 cm}\gamma_i \geq e^{\frac{b_i}{T\cdot W (1-\sum_{j}\frac{b_j}{T\cdot W\log(1+\gamma_j)})}}.\label{mtdopt2}
\end{align}
\end{subequations}

Given the scale and heterogeneity of the IoT, a centralized approach to find the optimal transmission probabilities will be challenging to implement in practice. This is due to the fact that it will require solving an optimization problem with very large number of variables and constraints. Thus, a centralized solution will always lead to high latency which can be significantly prohibitive for IoT services that are extremely sensitive to latency. Consequently, a self-organization approach to resource allocation is essential in the IoT as it allows the IoT devices to optimize their optimal transmission probabilities in a distributed manner. Based on optimization problems (\ref{HTDopt}) and (\ref{MTDopt1}), the optimal time fraction for each device $i$ is clearly dependent on the time fractions of the remaining devices, which motivates the use of a game-theoretic approach \cite{gametheory}, as detailed next.

\section{Heterogeneous time allocation game}

We formulate a static continuous noncooperative game defined by $(\mathcal{L}, (\mathcal{S}_i)_{i \in \mathcal{L}}, (U_i)_{i \in \mathcal{L}})$ with the players are being the devices in $\mathcal{L}$. The action $a_i$ of each HTD $i$ is to choose the time fraction $\tau_i$, and the action of each MTD $i$ is to find the required threshold $\gamma_i$ (which corresponds to finding its time fraction $\tau_i$ according to (\ref{thres})). Thus, given the constraints in (\ref{HTDopt})  and (\ref{mtdopt2}), respectively, the strategy space for each MTD $i$ is $\mathcal{S}_i=\Big[e^{\frac{b_i}{T\cdot W (1-\sum_{j \in \mathcal{M}, j \neq i}\frac{b_j}{T\cdot W\log(1+\gamma_j)})-\sum_{j \in \mathcal{H}\tau_j}}},\infty\Big]$ and for each HTD $i$ is $\mathcal{S}_i=[0, \min\{\frac{E}{ T P_i},1-\sum_{j \in \mathcal{H}, j \neq i}\tau_j-\sum_{j \in \mathcal{M}}\frac{b_j}{T\cdot W\log(1+\gamma_j)}\}]$. Consequently, the utility function of each MTD $i$ is
\begin{equation}
U_i(a_i)=-\frac{ P_i b_i}{ W\log(1+\gamma_i)e^{-\frac{\gamma_i \sigma^2}{\alpha_i^2P_i}}}.\label{utilmtd}
\end{equation}

The utility function of each HTD $i$ is
\begin{equation}U_i(a_i)=\mathbb{E}_{h_i}\Big[W\log(1+\frac{|h_i|^2P_i}{\sigma^2})\Big]  \tau_i.
\label{utilhtd}\end{equation} We can clearly see that the utility of each HTD $i$ in (\ref{utilhtd}) is increasing in $\tau_i$. The concavity of the utility each MTD $i$ is proved next.
\begin{lem}
	\emph{The utility of MTD $i$ in (\ref{utilmtd}) is concave in $\gamma_i$.}
\end{lem}
\begin{proof}
	By applying the $\log$ function to the energy function $\chi_i(\gamma_i)= \frac{ P_i b_i}{ W\log(1+\gamma_i)e^{-\frac{\gamma_i \sigma^2}{\alpha_i^2P_i}}}$ we have
	$\log \chi_i(\gamma_i)=\log P_ib_i-\log (W)-\log\log(1+\gamma_i)+ \frac{\gamma_i \sigma^2}{\alpha_i^2P_i}$. The term $-\log\log(1+\gamma_i)$ is convex in $\gamma_i$ since $\log \log(1+\gamma_i)$ is concave in $\gamma_i$. Also, the term $\frac{\gamma_i \sigma^2}{\alpha_i^2P_i}$ is linear in $\gamma_i$. Hence, the function $\log \chi_i(\gamma_i)$ is convex in $\gamma_i$. It follows that the energy function $\chi_i(\gamma_i)$ is convex in $\gamma_i$ as it logconvex in $\gamma_i$. Then, the utility of MTD $i$ is concave since $U_i(\gamma_i)=-\chi_i(\gamma_i)$.
\end{proof}
\begin{lem}\emph{By equating the derivative of  $U_i(\gamma_i)$ to zero, the utility function of MTD $i$ attains its maximum at $\gamma'_i$ that satisfies}
\begin{equation}
 \frac{1}{1+\gamma'_i}=\frac{\sigma^2}{\alpha_i^2P_i}\log(1+\gamma'_i).\label{mtdmax}
	\end{equation}
\end{lem}

In the proposed heterogeneous time allocation (HTA) game, the utility of each device $i$ based on (\ref{utilmtd}) and (\ref{utilhtd}) is dependent only on its own strategy. However, \emph{the strategy space of each device $i$ is dependent on the strategy vector $\boldsymbol{a}_{-i}$  of the remaining devices}. Hence, to find a suitable solution for this formulated game, we must study the concept of the generalized Nash equilibrium (GNE) \cite{gnep}.
\vspace{-0.2 cm}

\subsection{Generalized Nash Equilibrium Solution}

The GNE is a popular solution concept used to solve game-theoretic scenarios in which the action spaces of the players are mutually dependent, as is the case in the formulated HTA game. Formally, the GNE for the proposed game can be defined as follows.
\begin{defn}
	The \emph{GNE} of the heterogeneous time allocation game is the vector of players actions $\boldsymbol{a^*}$ such that $U_i(a^*_i) \geq U_i(a_i) \hspace{0.4 cm} \forall i \in \mathcal{D}, a_i \in \mathcal{S}_i(\boldsymbol{a}^*_{-i}).$
\end{defn}

In what follows, we show that the GNE for the HTA game always exists. Subsequently, we show the conditions under which the GNE is unique and characterize the GNE set in the case when the GNE is not unique.

\begin{prop}
\emph{A GNE exists for the heterogeneous time allocation game.}
\end{prop}
\vspace{-0.5 cm}
\begin{proof}
The proof is in Appendix A.
\end{proof}

In order to find the GNE, the best response of each device is derived as follows.
\begin{prop}
\emph{Given a strategy vector $\boldsymbol{a}_{-i}$, the best response of HTD $i$ is given by}
\begin{equation}
a_i=\min\Big({\frac{E_i}{ T P_i}, 1-\sum_{j \in \mathcal{H}, j \neq i}\tau_j}-\sum_{j \in \mathcal{M}}\frac{b_j}{T\cdot W\log(1+\gamma_j)}\Big)\label{brh1}
\end{equation}	\label{propGNEHTD}
\label{HTDbest}\end{prop}
\vspace{-1 cm}
\begin{proof}
The proof is in Appendix A.
\end{proof}

\begin{prop}
\emph{	The best response of each MTD $i$ is
	\[
	\hspace{-0.3 cm}a_i=
	\begin{dcases}
\gamma_{i,\textrm{UB}}, \textrm{if } \gamma'_i\geq\gamma_{i,\textrm{UB}}\\
	\gamma'_i,      \textrm{if} \hspace{0.1 cm} e^{\frac{b_i}{T\cdot W (1-\sum_{j \in \mathcal{M}, j \neq i}\frac{b_j}{T\cdot W\log(1+\gamma_j)}-\sum_{j \in \mathcal{H}}\tau_j)}} \leq \gamma'_i \leq \gamma_{i,\textrm{UB}}\\
	e^{\frac{b_i}{T\cdot W (1-\sum_{j \in \mathcal{M}, j \neq i}\frac{b_j}{T\cdot W\log(1+\gamma_j)}-\sum_{j \in \mathcal{H}}\tau_j)}}, \textrm{otherwise}
	\end{dcases}
	\]}\label{propGNEMTD}
\label{MTDbest}\end{prop}\vspace{-0.5 cm}
where $\gamma'_i$ is given by in (\ref{mtdmax}), and $\gamma_{i,\textrm{UB}}=-\frac{\alpha_i^2 P_i\log(1-\sqrt[t_i]{\epsilon})}{\sigma^2}$ is the upper bound on $\gamma$ as in (\ref{gammaub}).

\begin{proof}
The proof is in Appendix A.
\end{proof}
\vspace{-0.3 cm}
Thus, based on the properties of the best response, the GNE of both MTDs and HTDs is derived next. First, before characterizing the GNE of the proposed HTA game, we define the following key parameters

\begin{itemize}
\item The set $\mathcal{M}'=\Big \{i \in \mathcal{M} \hspace{0.3 cm}\textrm{s.t.} \hspace{0.3 cm} \gamma'_i \leq \gamma_{i,\textrm{UB}}  \Big \}$.
\item The set of sets
\\\small
 $\mathbb{A} =\{\mathcal{A}_j \subset \mathcal{L}$ s.t. $\sum_{i \in \mathcal{A} \cap \mathcal{M}'}\frac{b_i}{T\cdot W \log(1+\gamma'_i)}+\sum_{i \in \mathcal{A} \cap \mathcal{H}}\frac{E_i}{ T\cdot WP_i}+\sum_{i \in \mathcal{A}\setminus \mathcal{M}-\mathcal{M}'}\frac{b_i}{T\cdot W \log(1+\gamma_{i,\textrm{UB}})}\leq 1\}$.
 \end{itemize}
 \normalsize

%
%
%

%
 
Then, the GNE of the proposed HTA game can then be derived as follows. 
\begin{thm}
\emph{The GNE of the HTA game is dependent on the following two conditions:}
	\begin{itemize}
	\item \emph{If $\sum_{i \in \mathcal{M}'}\frac{b_i}{T\cdot W \log(1+\gamma'_i)}+\sum_{i \in \mathcal{H}}\frac{E_i}{ T\cdot WP_i}+ \sum_{i \in \mathcal{M}-\mathcal{M}'}\frac{b_i}{T\cdot W \log(1+\gamma_{i,\textrm{UB}})} \leq 1$, the GNE of heterogeneous time allocation game is \emph{unique}. In this case, the GNE strategy for each MTD $i$ is
	$a^{**}_i=\gamma'_i$ $\forall i \in \mathcal{M}'$ and $a^{**}_i=-\frac{\alpha_i^2 P_i\log(1-\sqrt[t_i]{\epsilon})}{\sigma^2}$ $\forall i \in \mathcal{M}-\mathcal{M}'$. For each HTD $i$, the GNE strategy is  $a^{**}_i=\frac{E_i}{ T\cdot WP_i}$.}
	
\item \emph{If $\sum_{i \in \mathcal{M}}\frac{b_i}{T\cdot W \log(1+\gamma'_i)}+\sum_{i \in \mathcal{H}}\frac{E}{ T\cdot WP_i}+\sum_{i \in \mathcal{M}-\mathcal{M}'}\frac{b_i}{T\cdot W \log(1+\gamma_{i,\textrm{UB}})}>1$, the GNE is \emph{not unique} and the GNE set $\mathcal{N}$ is given by
	$\mathcal{N}=\cup_{1\leq j \leq |\mathbb{A}|} \mathcal{N}_j$ where}
	
	$\mathcal{N}_j=\{\boldsymbol{a^*} \in \prod_{1\leq i \leq L}\mathcal{S}_i \hspace{0.3 cm} \textrm{s.t.} \hspace{0.3 cm} a^*_i=\gamma'_i \hspace{0.2 cm} \forall i \in \mathcal{A}_j \cap \mathcal{M}', a^*_i= \frac{E_i}{ T\cdot WP_i} \hspace{0.2 cm} \forall i \in \mathcal{A}_j \cap \mathcal{H}, a^*_i= -\frac{\alpha_i^2 P_i\log(1-\sqrt[t_i]{\epsilon})}{\sigma^2} \hspace{0.2 cm} \forall i \in \mathcal{A}_j \setminus \mathcal{M}-\mathcal{M}'$, \textrm{and} $\sum_{i \in \mathcal{L}-\mathcal{A}_j} a^*_i=1-\sum_{i \in \mathcal{A}_j \cap \mathcal{M}'}\frac{b_i}{T\cdot W \log(1+\gamma'_i)}+\sum_{i \in \mathcal{A}_j \cap\mathcal{H}}\frac{E_i}{ T\cdot WP_i}+ \sum_{i \in \mathcal{A}_j \cap (\mathcal{M}-\mathcal{M}')} \frac{b_i}{T\cdot W \log(1+\gamma_{i,\textrm{UB}})}$ s.t. $a^*_i < \gamma'_i$ $\forall i \in (\mathcal{L}-\mathcal{A}_j) \cap \mathcal{M}'$, $a^*_i < -\frac{\alpha_i^2 P_i\log(1-\sqrt[t_i]{\epsilon})}{\sigma^2} \hspace{0.1 cm} \forall i \in (\mathcal{L}-\mathcal{A}_j)\cap(\mathcal{M}-\mathcal{M}')$ and $a^*_i< \frac{E_i}{ T\cdot WP_i} \hspace{0.1 cm} \forall i \in (\mathcal{L}-\mathcal{A}_j)\cap \mathcal{H}\}.$
 \end{itemize}\label{GNEthm}
	\end{thm}
\begin{proof}
The proof is in Appendix B.
\end{proof}
	

%

\begin{algorithm}[t]
	\scriptsize The BS chooses an initial feasible vector of time allocations $\boldsymbol{\tau}_0$ and broadcasts it to its associated devices
	
	\scriptsize Each device initializes the current sum of time fractions as $\sum_{i \in \mathcal{L}}\tau_{i,0}$.
	
	\Repeat{convergence to GNE}{
		
		\ForEach{device $i \in 1, ..., L$}{
			Device $i$ computes its best response. 
			
			Device $i$ broadcasts its newly computed best response to all devices in $\mathcal{L}$.
			
			All devices update their current sum of time fractions. }
		
		
	}
	\caption{GNE Learning Algorithm for the Heterogeneous Time Allocation Game}\label{alg1}
\end{algorithm}

Based on the best response functions in Propositions \ref{HTDbest} and \ref{MTDbest}, it is clear that the game does not admit any dominant strategies. Thus, to find the GNE of the HTA game, we present a learning algorithm based on the nonlinear Gauss-Seidel type method in \cite{gnep}. This algorithm allows the IoT devices to find their GNE strategy in a distributed manner based on the best response dynamics. The algorithm is defined in Algorithm 1, and its complexity is derived next.

\vspace{-0.2 cm}

\begin{thm}
	\emph{The complexity of Algorithm \ref{alg1} is $O(L^2)$, and this algorithm converges in at most three iterations.}\label{complexityGNE}
\end{thm}
\vspace{-0.5 cm}
\begin{proof}
The proof is in Appendix B.
\end{proof}
\vspace{-0.4 cm}
\begin{cor}
\emph{The complexity of the number of computations done at each device, to find the GNE, is $O(L)$. }
\end{cor}
\vspace{-0.5 cm}
\begin{proof}
The proof is in Appendix B.
\end{proof}

The GNE assumes that players are fully rational which implies that they have enough computational powers to compute their GNE and can gather precise information on the actions of other devices. Indeed, to compute the GNE, each IoT device will need to record the actions of all other IoT devices. This can lead to a massive amount of information exchange in each iteration among the IoT devices and also requires all IoT devices to have enough memory to store all of these actions. This may not be realistic in an IoT ecosystem in which devices are heterogeneous and have different computational capabilities in terms of memory and processing powers. For instance, IoT devices can range from small sensors and wearable devices that have very limited computational capabilities to smartphones that have higher computational capabilities. To cater for such constraints, an alternative approach to IoT resource allocation is proposed next. This approach extends the powerful framework of \emph{cognitive hierarchy theory} \cite{Camerer_acognitive}, that can properly accommodate such heterogeneous device capabilities. 
\vspace{-0.4 cm}
\subsection{Cognitive Hierarchy Framework for IoT Resource Allocation}
Cognitive hierarchy theory \cite{Camerer_acognitive} is a branch of behavioral
game theory based on the concept of bounded rationality.
In general, bounded rationality means that each player finds the
best strategy based on the information that is accessible to this player
and the player's computational or cognitive capacity as well as the time available for decision making.
The CH framework considers a hierarchy of players in which players are distributed into discrete levels of rationality. Here, players have different beliefs about other players depending on their rationality level. The beliefs formed by each player at each rationality level is dependent on the information that the device can obtain and store, which is constrained by the device's memory capacity and the resources available at the device to obtain the information (such as energy). A player at level $k$ selects its strategy based on the strategies of players belonging to levels lower or equal to $k$. Players belonging to the lowest level $0$ do not engage in any rational thinking and, instead, they select their strategy randomly. We note that CH significantly differs from classical hierarchical games, such as Stackelberg games and their variants \cite{gametheory}, in which all players are considered to have the same rationality.  Also, in a Stackelberg game, the hierarchy levels are defined based solely on the roles of the players, i.e., players are classified as either leaders or followers, rather than based on capabilities, as is done in CH. 


For our problem, devices are restricted by their computational capabilities and resources. A player at a higher level of rationality can consider more levels when computing its strategy. Hence, devices must be grouped into multiple levels of intelligence depending on their capabilities. Thus, we propose a CH model that extends the model of \cite{Camerer_acognitive} and which has the following characteristics: \begin{enumerate}
	\item The number of players at each level $k$ is distributed according to a Poisson distribution $f$ with rate $\tau$.
	Since MTDs generally have lower computational capabilities than HTDs, it is assumed that MTDs belong to lower rationality levels than HTDs. Hence, MTDs belong to level 0 up to some level $l$ while HTDs belong to rationality levels greater than $l$. Also, we assume that devices belonging of the same CH level are of the same type.
	\item A player at level $k$ knows the true proportions $f(0)$, $f(1)$,...,$f(k-1)$ of players which are at lower rationality levels. Since these proportions do not add up to one, a player at level $k$ computes the relative frequency $g_k(h)$ of players at a level $h$ ($0 \leq h < k$) as
	$g_k(h)=\frac{f(h)}{\sum_{i=0}^{k-1}f(i)}$, 
and $g_k(h)=0$ $\forall h \geq k$.
	This assumption is known as the overconfidence assumption and is a good model for games with human players. However, in our case, the players are IoT devices, and it is more likely to have devices of the same computational capabilities. Thus, we relax this assumption so that a device at level $k$ also knows the proportion of devices at the same level $k$. This extension to CH is challenging because the strategy of each device will depend on the strategies taken by devices at the same level thus requiring new solutions that go beyond the existing literature [12].
	\item In a classical CH theory model \cite{Camerer_acognitive}, players at the lowest level 0 are assumed to make their choices randomly according to a uniform distribution. In the context of IoT, level-0 devices correspond to MTDs with very limited resources such as sensors. However, a random choice of $\gamma_i$ by a level-0 device $i$ according to a uniform distribution might result in occupying a large fraction of the time duration, which will considerably degrade the performance of the remaining devices. Thus, we consider that each device at level 0 will choose its action randomly within $[\tau_{0,LB},1]$ according to a decreasing distribution such as an exponential distribution with mean $\mu$. The time fraction $\tau_{0,LB}$ is the lower bound on the time fraction of level-0 devices according to (\ref{MTDopt1}).
\end{enumerate}

Note that the Poisson distribution has been shown in \cite{Camerer_acognitive} to accurately capture situations in which fewer players will be at a level higher than $k$, as the rationality level $k$ grows larger. As a result, in our setting, this assumption holds because there is a limit on the computational capability of the devices (in terms of memory and processing power). Indeed, devices with higher computational capability are more expensive, and, hence, are fewer in number. In a typical IoT ecosystem, the number of HTDs is small as opposed to the numerous sensors pertaining to different types of IoT applications.


Due to the aforementioned CH characteristics, each player will seek to find its CH strategy using an iterative process. During this process, each player will perform limited steps of strategic thinking depending on its rationality level. A device at level $k$ will anticipate the strategies of devices at levels $0$ to $k-1$. Given that the strategies of lower level devices will help a level-$k$ device to take a more informed decision about the IoT network, then, this device performs $k$ steps of thinking.


In the HTA game, the strategy space of each device $i$ is dependent on the actions of the remaining devices. Further, in CH, each device has its own beliefs $g_k$ about the rationality (and actions) of the remaining devices depending on its rationality level $k$. Consequently, the strategy space of each device $i$ is dependent on the device's beliefs $g_k$. Constraint $\sum_{i \in \mathcal{L}} \tau_i=1$ becomes $\mathbb{E}_{g_k}\Big[\sum_{i \in \mathcal{L}} \tau_i\Big]=1$. We denote by $\mathcal{S}_{i,g_k}$ the strategy space of device $i$ based on its belief $g_k$. The expression of $\mathbb{E}_{g_k}\Big[\sum_{i \in \mathcal{L}} \tau_i\Big]$  of an MTD $i$ at a CH level $k$ ($1 \leq k \leq l-1$) is given by
\begin{equation}
\mathbb{E}_{g_k}\Big[\sum_{i \in \mathcal{L}} \tau_i\Big]=\sum_{h=0}^{k}g_k(h)\sum_{j \in \mathcal{L}}\frac{b_h}{T\cdot W\log(1+\gamma^*_j(h))}\label{expectedm}
\end{equation}
The expression of $\mathbb{E}_{g_k}\Big[\sum_{i \in L} \tau_i\Big]$  of an HTD $i$ at a CH level $k$ ($l \leq k$) is given by
\begin{eqnarray}
\mathbb{E}_{g_k}\Big[\sum_{i \in \mathcal{L}} \tau_i\Big]&=&\sum_{h=0}^{l-1}g_k(h)\sum_{j \in \mathcal{L}}\frac{b_h}{T\cdot W\log(1+\gamma^*_j(h))}\nonumber\\
&&+\sum_{h=l+1}^{k}g_k(h)\sum_{j \in \mathcal{L}}\tau^*_j(h)\label{expectedh}
\end{eqnarray}
where $\gamma^*_j(h)$ is the CHE strategy of MTD $j$ at level $h$ and $\tau^*_j(h)$ is the CHE strategy of HTD $j$ at level $h$.  Since devices belonging to the same CH level are of the same type, the only parameter that is different is the channel quality.
 Computing (\ref{expectedm}) and (\ref{expectedh}) requires evaluation of the CHE strategy at each level $h$ ($h\leq k$) for all the channel variance values of the devices in $\mathcal{L}$, which has a computational complexity linear in the number of devices $L$. Given that the number of devices $L$ is very large, we assume that the BS quantize the set of channel variances of all devices into a discrete set $\mathcal{C}$ of $C$ values where $C << L$. Then, it broadcasts the quantized set to all of its associated devices.
 Hence, (\ref{expectedm}) and (\ref{expectedh}) can be rewritten as:
\begin{equation}
\mathbb{E}_{g_k}\Big[\sum_{i \in \mathcal{L}} \tau_i\Big]=\sum_{h=0}^{k}g_k(h)\sum_{q \in \mathcal{C}}N_q\frac{b_h}{T\cdot W\log(1+\gamma^*_q(h))},\label{expectedmq}
\end{equation}
\begin{eqnarray}
\mathbb{E}_{g_k}\Big[\sum_{i \in \mathcal{L}} \tau_i\Big]&=&\sum_{h=0}^{l-1}g_k(h)\sum_{q \in \mathcal{C}}N_q\frac{b_h}{T\cdot W\log(1+\gamma^*_q(h))} \nonumber\\
&&+\sum_{h=l+1}^{k}g_k(h)\sum_{q \in \mathcal{C}}N_q\tau^*_q(h),\label{expectedhq}
\end{eqnarray}
where  $\gamma^*_q(h)$ is the CHE of MTD having quantized channel variance value $q$ at CH level $h$, and $\tau^*_q(h)$ is the CHE of HTD having quantized channel variance value $q$ at CH level $h$, $N_q$ is the number of devices having quantized channel variance value $q$.
To compute (\ref{expectedmq}) and (\ref{expectedhq}), we can see that each device $i$ at CH level $k$ must know the expected equilibrium strategies of devices at the same level. This is challenging since each device finds its CHE strategy based on its beliefs and not by using learning based on the devices' actions as in GNE. Since devices belonging to the same level are of the same type, we assume that each device $i$ believes that all devices at the same rationality level will choose the same action.

Next, we define the cognitive hierarchy equilibrium of the HTA game as follows.

%


\begin{defn}
	A strategy profile $\boldsymbol{a}^*$ is said to constitute a \emph{ cognitive hierarchy equilibrium} for the HTA game if and only if:
	\vspace{-0.1 cm}
	\begin{equation}
	\boldsymbol{a}^*_i(k)= \arg \max_{\boldsymbol{a}_i \in \mathcal{S}_{i,g_k}} U_{i}(\boldsymbol{a}_i) \hspace{0.1 cm}, \forall i \in \mathcal{L},
	\vspace{-0.1 cm}
	\end{equation}
	where $k$ is the CH level of player $i$ and $\mathcal{S}_{i,g_k}$ is the strategy space of device $i$ based on its belief $g_k$.
\end{defn}


To find the CHE, MTD $i$ at CH level $k \geq 1$ solves the following problem
\small
\begin{subequations}
\begin{align}
&\min_{\gamma_i}\frac{  P_k b_k}{W\log(1+\gamma_i(k))e^{-\frac{\gamma_i(k) \sigma^2}{\alpha_i^2P_k}}},\\
& \textrm{s.t.} \hspace{0.1 cm}\Big(1-e^{-\frac{\gamma_i(k) \sigma^2}{\alpha_i P_k}}\Big)^{t_k} \leq \epsilon, \\
& \frac{b_k}{T\cdot W\log(1+\gamma_i(k))}\leq \frac{1}{g_k(k) \cdot L} \Big(1-\sum_{h=0}^{k-1}g_k(h)\nonumber\\&\hspace{3 cm}\times\sum_{q \in \mathcal{C}}N_q\frac{b_h}{T\cdot W\log(1+\gamma^*_q(h))}\Big).
\end{align}\label{CHMTDopt2}
\end{subequations}
\normalsize
A level $k$ HTD $i$ solves the following optimization problem
\vspace{0.2 cm}
\small
 \begin{subequations}
\begin{align}
&\max_{\tau_i(k)} \mathbb{E}_{h_i}[W\log(1+\frac{|h_i|^2P_k}{\sigma^2})] \tau_i(k),\\
& \textrm{s.t.}\hspace{0.3 cm} P_k \tau_i(k) \leq \frac{E_i}{ T}, 0 \leq \tau_i(k) \leq 1,\\
& \small\hspace{-0.1 cm} \tau_i(k) \leq \frac{1}{g_k(k) \cdot L}\Big(1-\sum_{h=l+1}^{k-1}g_k(h)\sum_{q \in \mathcal{C}}N_q\tau^*_q(h)\nonumber\\
&\hspace{1 cm}-\sum_{h=0}^{l}g_k(h)\sum_{q \in \mathcal{C}}N_q\frac{b_h}{T\cdot W\log(1+\gamma^*_q(h))}\Big). 
\end{align}\label{CHHTDopt}
\vspace{-0.1 cm}
\end{subequations}

\normalsize
Based on the optimization problems in (\ref{CHMTDopt2}) and  (\ref{CHHTDopt}), to find its optimal CHE strategy, a device $i$ at CH level $k$ needs to compute the actions of the remaining devices for each CH level $h$ ($0 \leq h < k$) based on its belief $g_k$.  Hence, in CH, each device finds its CHE based on its own beliefs about other devices and not through an iterative learning process as in GNE. Thus, the CHE is a stable solution (similar to the stability of a Nash equilibrium) since each player, due to its bounded rationality, has no incentive of deviating to another strategy once it computes its CHE. Therefore, in terms of stability, the CHE provides the same stability as the GNE, but under the more realistic model with bounded rationality. Further, the CHE strategy of a device at level $k$ is a function of the channel quality between the BS and the device. Thus, devices belonging to the same CH level might have different CHE strategies. In the following proposition, we characterize the CHE of the HTA game.
\begin{prop}
\emph{The CHE of the HTA game is the vector $\boldsymbol{a}^*$ such that for MTD $i$ at level $k \geq 1$ the CHE strategy is given by (\ref{propeq}),
			and the CHE strategy of HTD $i$ at level $k$ is \small
			\\
			$a^*_i(k)=\min \Big \{ \frac{E_i}{ T P_i}, \frac{1}{g_k(k) \cdot L}\Big(1-\sum_{h=l+1}^{k-1}g_k(h)\sum_{q \in \mathcal{C}}N_q\tau^*_q(h)-\sum_{h=0}^{l}g_k(h)\sum_{q \in \mathcal{C}}N_q\frac{b_h}{T\cdot W\log(1+\gamma^*_q(h))}\Big) \Big \}.$
		}\normalsize
\end{prop}
\newcounter{tempEquationCounter1} 
\newcounter{thisEquationNumber1}
\newenvironment{floatEq1}
{\setcounter{thisEquationNumber1}{\value{equation}}\addtocounter{equation}{1}
\begin{figure*}[!t]
\normalsize\setcounter{tempEquationCounter1}{\value{equation}}
\setcounter{equation}{\value{thisEquationNumber1}}
}
{\setcounter{equation}{\value{tempEquationCounter1}}
\hrulefill\vspace*{4pt}
\end{figure*}
}
\begin{floatEq1}
\begin{equation}
a^*_i(k)=\left\{
			\begin{array}{ll}
			\gamma_{i,\textrm{UB}}, \hspace{0.2 cm}\textrm{if } \gamma'_i\geq \gamma_{i,\textrm{UB}},\\
			\gamma'_i,   \textrm{if} \hspace{0.1 cm} e^{\frac{b_k}{\frac{T\cdot W}{g_k(k)} \Big(1-\sum_{h=0}^{k-1}g_k(h)\sum_{q \in \mathcal{C}}N_q\frac{b_h}{T\cdot W\log(1+\gamma^*_q(h))}\Big)}}-1\leq \gamma'_i \leq \gamma_{i,\textrm{UB}},\\
			e^{\frac{b_k}{\frac{T\cdot W}{g_k(k) \cdot L} \Big(1-\sum_{h=0}^{k-1}g_k(h)\sum_{q \in \mathcal{C}}N_q\frac{b_h}{T\cdot W\log(1+\gamma^*_q(h))}\Big)}}-1, \hspace{0.1 cm} \textrm{otherwise}.
			\end{array} \right.
\label{propeq}
\end{equation}
\end{floatEq1}
\begin{proof}
The proof is in Appendix A.
\end{proof}
Given this characterization, we can compute the computational complexity of finding the CHE at each CH level.
\begin{thm}
\emph{The complexity of the optimization done at device of level $k$ to find the CHE is $O((k+1)\cdot C)$ where $C$ is the size of the quantized channel variance set $\mathcal{C}$}.\label{propCHcomplexity}
\end{thm}
\vspace{-0.3 cm}
\begin{proof}
The proof is in Appendix B.
\end{proof}


Theorem \ref{propCHcomplexity} first shows that the process of finding the CHE at each CH level $k$ is guaranteed to converge. Further, Theorem 3 shows that the computational complexity of the optimization done by each device at CH level $k$ to find the CHE is linear in the CH level while the computational complexity of the GNE learning algorithm is polynomial in the number of devices $L$. Since $k << L$,  Theorem \ref{propCHcomplexity} validates that the CH proposed approach is more suitable for IoT devices that are of heterogeneous computational capabilities and typically resource constrained.
	
	
	
In the proposed CH approach, all devices obtain their CHE strategy based on their beliefs. This is reflected in the constraint $\mathbb{E}_{g_k}\Big[\sum_{i \in \mathcal{L}} \tau_i\Big]=1$. Hence, the resulting sum of time fractions of all devices at CHE may not be equal to 1. However, when $\sum_{i \in \mathcal{L}}\tau_{i,\textrm{CHE}}\geq 1$, the time fractions of all devices must be normalized. The sum of time fractions at CHE is given by
\vspace{-0.1 cm}
\begin{equation}
\sum_{i \in \mathcal{L}}\tau_{i,\textrm{CHE}}=\sum_{i \in \mathcal{M}}\frac{b_i}{ T\cdot W\log(1+a^*_i(h_i))}+\sum_{i \in \mathcal{H}}a^*_i(h_i)
\end{equation}
\vspace{-0.1 cm}
where $\tau_{i,\textrm{CHE}}$ is the value of time fraction of device $i$ at CHE and $h_i$ is the CH level of device $i$.

Let $\gamma_{g_k}=e^{\frac{b_k}{\frac{T\cdot W}{g_k(k)} \Big(1-\sum_{h=0}^{k-1}g_k(h)\sum_{q \in \mathcal{C}}N_q\frac{b_h}{T\cdot W\log(1+\gamma^*_q(h))}\Big)}}-1$.

 The normalized time fraction for each MTD $i$ will be
\small
 \begin{equation}	\hspace{-0.6 cm}\nu_i=\left\{
	\begin{array}{ll}
	\frac{\frac{b_i}{ T\cdot W\log(1+\gamma_{i,\textrm{UB}})}}{\sum_{i \in \mathcal{M}}\frac{b_i}{ T\cdot W\log(1+a^*_i(h_i))}+\sum_{i \in \mathcal{H}}a^*_i(h_i)}, \hspace{0.2 cm}\textrm{if } \gamma'_i\geq \gamma_{i,\textrm{UB}},\\
\frac{\frac{b_i}{ T\cdot W\log(1+\gamma'_{i})}}{\sum_{i \in \mathcal{M}}\frac{b_i}{ T\cdot W\log(1+a^*_i(h_i))}+\sum_{i \in \mathcal{H}}a^*_i(h_i)},  \hspace{0.2 cm}    \textrm{if} \hspace{0.2 cm} \gamma_{g_k}\leq \gamma'_i \leq \gamma_{i,\textrm{UB}},\\
	\frac{\frac{b_k}{\frac{T\cdot W}{g_k(k) \cdot L} \Big(1-\sum_{h=0}^{k-1}g_k(h)\sum_{q \in \mathcal{C}}N_q\frac{b_h}{T\cdot W\log(1+\gamma^*_q(h))}\Big)}}{\sum_{i \in \mathcal{M}}\frac{b_i}{ T\cdot W\log(1+a^*_i(h_i))}+\sum_{i \in \mathcal{H}}a^*_i(h_i)} \hspace{0.2 cm} \textrm{otherwise}.
	\end{array} \right. 
\end{equation}
\normalsize	
The normalized time fraction for each HTD $i$ is given by (25).


%



\newcounter{tempEquationCounter} 
\newcounter{thisEquationNumber}
\newenvironment{floatEq}
{\setcounter{thisEquationNumber}{\value{equation}}\addtocounter{equation}{1}
\begin{figure*}[!t]
\normalsize\setcounter{tempEquationCounter}{\value{equation}}
\setcounter{equation}{\value{thisEquationNumber}}
}
{\setcounter{equation}{\value{tempEquationCounter}}
\hrulefill\vspace*{4pt}
\end{figure*}

}

\begin{floatEq}
\begin{equation}
\hspace{-0.5 cm}\nu_i=\frac{\min \Big \{ \frac{E}{ T P_i}, \frac{1}{g_k(k) \cdot L}\Big(1-\sum_{h=l+1}^{k-1}g_k(h)\sum_{q \in \mathcal{C}}\tau^*_q(h)-\sum_{h=0}^{l}g_k(h)\sum_{q \in \mathcal{C}}N_q\frac{b_h}{T\cdot W\log(1+\gamma^*_q(h))}\Big) \Big \}}{\frac{b_i}{ T\cdot W\log(1+a^*_i(h_i))}+\sum_{i \in \mathcal{H}}a^*_i(h_i)}. 
\label{equ:floatedEquation}
\end{equation}
\end{floatEq}

We define $\gamma_{\nu_i}$ to be the SNR threshold for device $i$ when the time fraction of device $i$ is $\nu_i$.

In practice, for the first two transmission time slots, the devices transmit according to their computed CHE time fraction. In our proposed scheme, devices transmit in a pre-determined order, e.g. device 1, device 2, until device $L$. Since the IoT devices are densely deployed, device $i+1$ can hear and decode the packet transmitted by device $i$ to the BS. Device $i+1$ can detect the signal transmitted by device $i$ either by carrier sensing or by energy detection, i.e. if the energy of the transmitted signal exceeds a predefined threshold, as done for example in CSMA protocols \cite{CSMA}. Also, device $i+1$ can decode the packet and determine from its header if the packet is transmitted by device $i$. Thus, device $i+1$ starts transmitting as soon as device $i$  completes its transmission. Then, each device records the times $t_{i,1}$ and  $t_{i,2}$ during which it accesses the channel in the first and second time slot, respectively.  Finally, each device $i$ computes the time duration $T'$ as the difference $t_{i,2}-t_{i,1}$. Finally, device $i$ normalizes its time fraction as $\nu_i=\frac{\tau_{i,CHE}\cdot T}{t_{i,2}-t_{i,1}}$. 

In the proposed normalization process,  each device computes its normalized time fraction based on the transmission times of the remaining IoT devices. Thus, transmission is not interrupted during the process of normalization. Further, there are no extra control messages required to perform the normalization. Such a normalization scheme does not require information exchange among the devices, and, thus, it does not result in signaling overhead. Hence, it is suitable for dynamic networks where devices join or leave the network.

In the following proposition, we compare the performance of the CHE and GNE strategy profiles of the HTA game. We show the conditions under which the CHE strategy can constitute also a GNE strategy, and show that in the other cases (i.e. when the CHE strategy profile does not constitute a GNE) that there always exists at least a GNE that yields a better performance, in terms of the total system utility. Naturally, such a performance improvement stems from the fact that the GNE requires more information and more computations per device. 

\begin{thm}
	\emph{Let $\boldsymbol{d}=(d_1,d_2,...,d_{L})$ be a strategy profile corresponding to the normalized CHE time fractions with $d_i=\nu_i$ if device $i$ is an HTD and $d_i=\gamma_{\nu_i}$ if device $i$ is an MTD 
		\begin{itemize}
		\item If $\sum_{i \in \mathcal{L}}\tau_{i,\textrm{CHE}} \geq 1$, the strategy vector $\boldsymbol{d}$ is a GNE strategy of the HTA game if for every level-0 device $i$, its CHE strategy satisfies  $a^*_i(0) \geq \min\{\gamma'_i,\gamma_{i,\textrm{UB}}\}$.
		\item If $\sum_{i \in \mathcal{L}}\tau_{i,\textrm{CHE}} < 1$, the CHE strategy vector of the HTA game is a GNE if
		\begin{itemize}\item $\sum_{i \in \mathcal{M}'}\frac{b_i}{T\cdot W \log(1+\gamma'_i)}+\sum_{i \in \mathcal{H}}\frac{E_i}{ T\cdot WP_i}+ \sum_{i \in \mathcal{M}-\mathcal{M}'}\frac{b_i}{T\cdot W \log(1+\gamma_{i,\textrm{UB}})} \leq 1$,
             \item The CHE equilibrium strategy for each MTD $i$ is $a^*_i(k_i)=\gamma'_{i}$ if $i \in \mathcal{M}'$ and $a^*_i(k_i)=\gamma_{i,\textrm{UB}}$ if $i \in \mathcal{M}-\mathcal{M'}$ $\forall k_i \geq 0$.
             \item The CHE equilibrium strategy for each HTD $i$ is $a^*_i(k_i)=\frac{E_i}{ T\cdot WP_i}$.
		\end{itemize}
		\hspace{0.2 cm}Otherwise, there exists at least one GNE that yields a better performance than the CHE, in terms of the total system utility.
		\end{itemize}}
\label{theoremcompare}\end{thm}
\vspace{-0.5 cm}
\begin{proof}
The proof is in Appendix B.
\end{proof}
Theorem \ref{theoremcompare} provides the conditions under which the strategy profile $\boldsymbol{d}$ corresponding to the normalized CHE time fractions constitutes a GNE strategy. Hence, the devices do not lose in performance by having limited information and computational capabilities. However, if these conditions are not met, the CHE strategy does not constitute, in most cases, a GNE strategy, and the performance is degraded due to bounded rationality. 
\vspace{-0.5 cm}
\subsection{Signaling Overhead}

The communication overhead is specifically evaluated in terms of the total bits exchanged to reach the CHE and GNE solutions, respectively. In the proposed CH scheme, the BS needs to broadcast one packet to its associated devices whenever devices enter or leave the network (e.g., during handshaking or device registration). The broadcasted packet includes the updated network information necessary for the devices to compute the CHE. Thus, the broadcasted packet will include the number of devices $N$, the quantized channel variance set $\mathcal{C}$, and the proportion of devices $f_q$ having quantized channel value $q \in \mathcal{C}$. The BS acquires the CSI of each device upon registration to the network. Further, the packet includes the parameters' values and QoS constraints of each type of devices. For each type $h$ of MTDs, the packet includes the packet size $b_h$ and the time constraint $t_h$. For each type $h$ of HTDs, the packet includes the energy constraint $E_h$.

Thus, the size of the broadcasted packet depends on the number of bits required to represent each of the aformentioned parameters. Concerning the number of devices $N$, the number of bits required $B_N$ depends on $N_{\max}$, the maximum number of devices that the BS can serve. Assuming that $N_{\max}$ is 10,000, the required number of bits will be  $B_N=14$. As for the channel information and considering a static setting, the channel variance value $\alpha_q$ is typically less than one, and the number of required bits $B_{\alpha}$ is 4 bits assuming that the fractional part of the quantized variance value is represented by 1 digit. Also, the proportion of devices having a quantized variance value $\alpha_q$ can be represented by $B_f=7$ assuming that the fractional part of each proportion is captured by 2 digits. Next, for each type of MTDs, the number of required bits for the packet size is $B_b=10$ since MTDs usually have small packet sizes ($<$ 1,000 bits). Assuming that the maximum delay is $1,000$ ms and that the time slot duration is $1$ ms \cite{IEEEpacket, IEEEpacket2}, the number of bits $B_t$ needed to represent the time constraint of each type of MTDs is $B_t=10$. For HTDs, assuming that the energy upper bound $E_h$ (in $\mu$ J) is represented by two digits, the number of bits required will be $B_E=7$. Hence, the number of bits to represent each type of MTD is $B_M=B_b+B_t=20$ and for each type of HTD is $B_H=B_E=7$. The size of the packet (in bytes) is 
$M=\frac{14C + 7 N_H+ 20 N_M}{8}$ where $N_H$ and $N_H$ are the number of types of HTDs and MTDs respectively.
For a network with two types of MTDs and one type of HTDs and when $C=5$, the size of the packet will be $15$ bytes. The typical packet size of an MTD with time critical application [23] ranges from $40$ to $1,000$ bytes. Thus, the proposed CH scheme does not incur significant overhead.

As for the GNE, the overhead stems mainly from the actions exchanged among the devices to obtain the GNE.
In Theorem 2, it is shown that, in the worst case, three rounds of communication are required among the devices to converge to the GNE. Assuming that the fractional part of the time fraction is represented by $3$ digits, the number of bits required to represent the time fraction is $10$. Thus, the total number of bits exchanged to reach the GNE is $30 \cdot L$ . When the number of devices is 1,000, the total number of bits exchanged will be $30,000$. Hence, the overhead of the GNE solution is significantly higher than the proposed CHE.

\vspace{-0.3 cm}
\section{Simulation Results and Analysis}

In our simulations, we set the bandwidth $W$ to $100$ MHz, the time period $T$ to $3$ ms, the noise variance $\sigma^2$ to $-90.8$ dBm, and $\epsilon$ to $0.0001$. The value of the time period is chosen to be small enough to be suitable for IoT devices that have ultra low latency requirements, yet adequate to accomodate the minimum time requirements for an IoT with massive number of devices. We consider three types of devices: the first two are MTDs while the third is HTD. MTD type 1 devices represent MTDs with strict latency requirements such as e-healthcare sensors while MTD type 2 devices represent MTDs with relaxed delay constraints such as smart meters. Thus, the packet size for the MTD types 1 and 2 are set 128 and 50 bytes respectively, and the latency constraint for MTD types 1 and 2 are set to be $5$ ms and $1$ s respectively in line with the guidelines of \cite{IEEEpacket, IEEEpacket2, m2mpacket, 5Gpaper}. The transmission powers for MTD types 1 and 2 and for HTDs are 0.1, 0.1, and 0.5 W respectively. The variance of the Rayleigh fading distribution of all channels is $\alpha^2=0.1$. The rate of the Poisson distribution of devices over the CH levels is $\tau=1$. This value implies that the proportion of devices in each CH decreases with the CH levels, and the probability that a device belongs to a CH level higher than 3 is negligible. Thus, it is suitable for our system because the maximum number of CH levels considered is 3.

First, we assess the performance of the GNE solution of the HTA game by computing the \emph{price of anarchy (PoA)}  \cite{gametheory} as follows.
\vspace{-0.6 cm}
\subsection{Price of Anarchy of the GNE}
The PoA is the ratio between the optimal centralized solution that maximizes total utility, and the minimum total utility that can be possibly achieved by a GNE strategy. Thus, the PoA is a measure of how much the performance of the system degrades when the GNE is used. In our problem, the utility of HTDs is expressed in terms of the achieved rate whereas the utility of each MTD is expressed in terms of the energy cost. Hence, in order to have a better evaluation of the performance of MTDs and HTDs, we compute separate PoAs for the MTDs and for the HTDs. The PoA of MTDs is defined as the ratio of the maximum energy consumed that can be possibly achieved by a GNE and the minimum possible energy consumed by MTDs. The PoA of HTDs, on the other hand, is defined as the ratio of maximum total rate of HTDs and the minimum total rate that can be possibly achieved by a GNE solution. 
%
%
\begin{figure}[t]
\minipage{0.45\textwidth}
  \includegraphics[width=8 cm,height=5cm,angle=0]{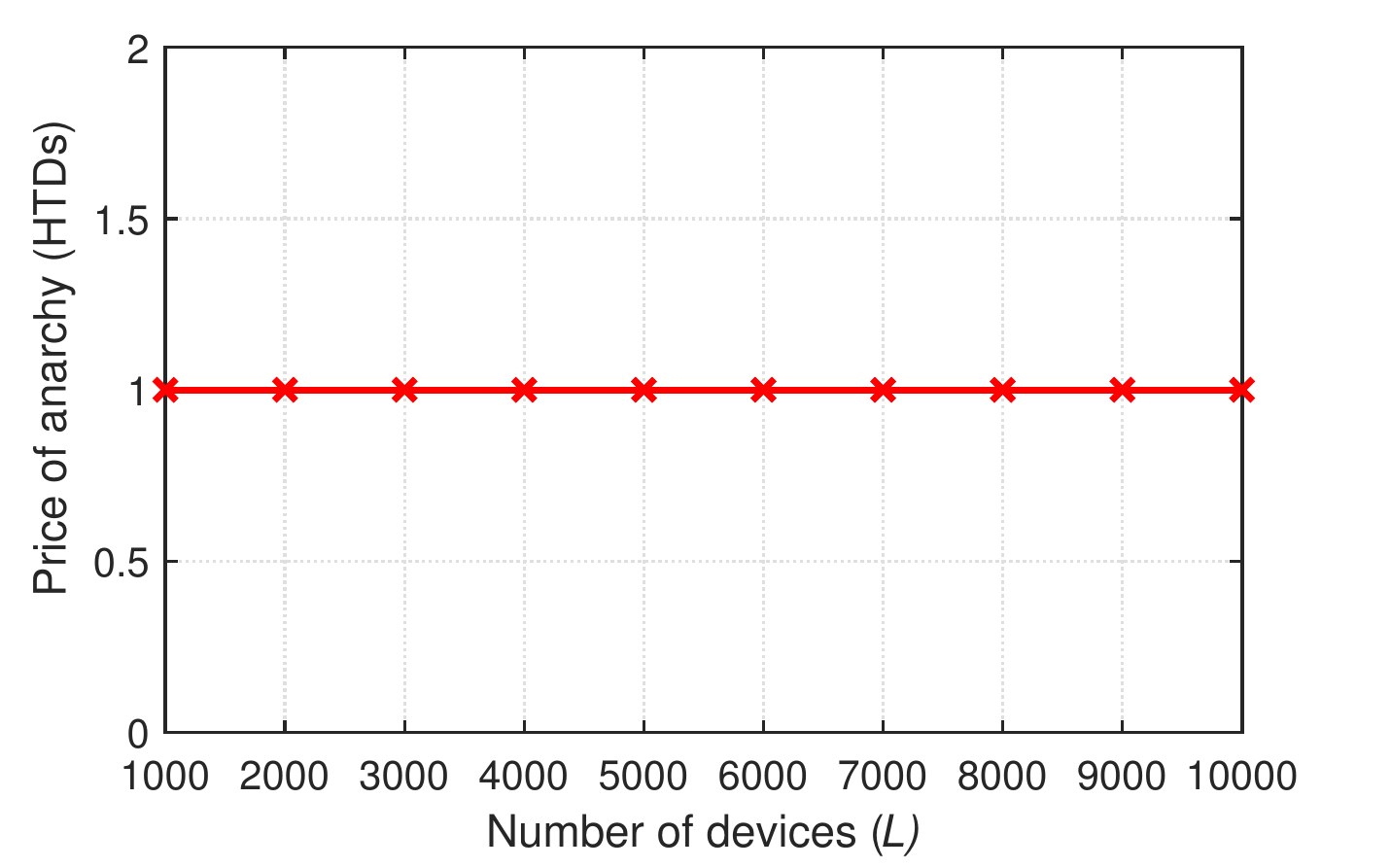}
  \caption{Price of anarchy of HTDs vs. number of devices.}\label{PoAh}
\endminipage \hspace{1 cm}
\minipage{0.45\textwidth}
  \includegraphics[width=8 cm,height=5cm,angle=0]{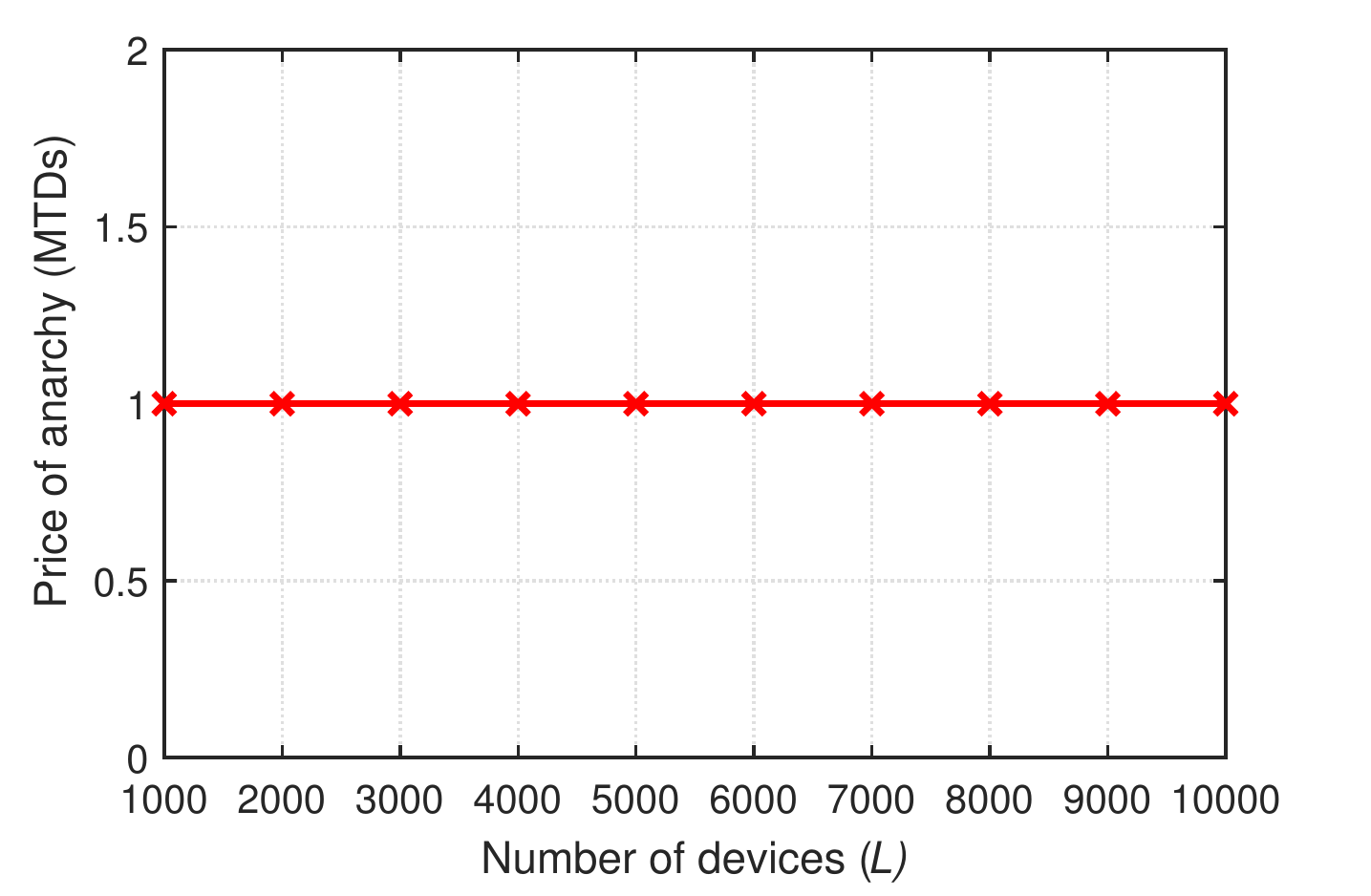}
  \caption{Price of anarchy of MTDs vs. number of devices.}\label{PoAm}
\endminipage
\vspace{-0.8 cm}
\end{figure}
Fig. \ref{PoAh} shows the PoA of HTDs as a function of the network size. As shown in Fig. \ref{PoAh}, the PoA of HTDs is one for the considered network sizes since the GNE is unique in this case. Fig. \ref{PoAm} shows the PoA of MTDs as a function of the network size. Again, the PoA of MTDs is also $1$ for the considered network sizes since the GNE is unique. Thus, Figs. \ref{PoAh} and  \ref{PoAm} show that the GNE solution provides a stable performance for both MTDs and HTDs.
\vspace{-0.4 cm}

\subsection{CHE Solution Evaluation}

For the considered simulation values, we compute the CHE solution when the network size varies between 1,000 and 10,000 in steps of 1,000.  Since level-0 devices choose their time fraction randomly, 1,000 samples of level-0 devices time fractions are generated, and the total sum of CHE time fractions is computed. Then, the average sum CHE time fraction as well as the average time fraction of each MTD type 2 and each HTD are computed.  We consider two cases: A first case in which the mean $\mu$ of the distribution of time fractions of level-0 devices is $2\tau_{0,LB}$ where $\tau_{0,LB}$ is lower bound on the time fractions of level-0 devices and a second case with $\mu=3\tau_{LB}$. Fig. \ref{sumtau} shows  the average sum of CHE time fractions versus the network size. As shown in Fig.  \ref{sumtau}, the average sum of CHE time fractions increases with the network size where it exceeds one for network sizes greater than $8,000$. When $\mu$ increases to $3\tau_{0,LB}$, the sum of CHE time fractions increases for each considered network. This is because, as $\mu$ increases, level-0 devices have a higher chance of transmitting with larger time fractions, which eventually increases the sum of CHE time fractions.

\begin{figure}[t]
	\centering
	\includegraphics[width=8 cm,height=5cm,angle=0]{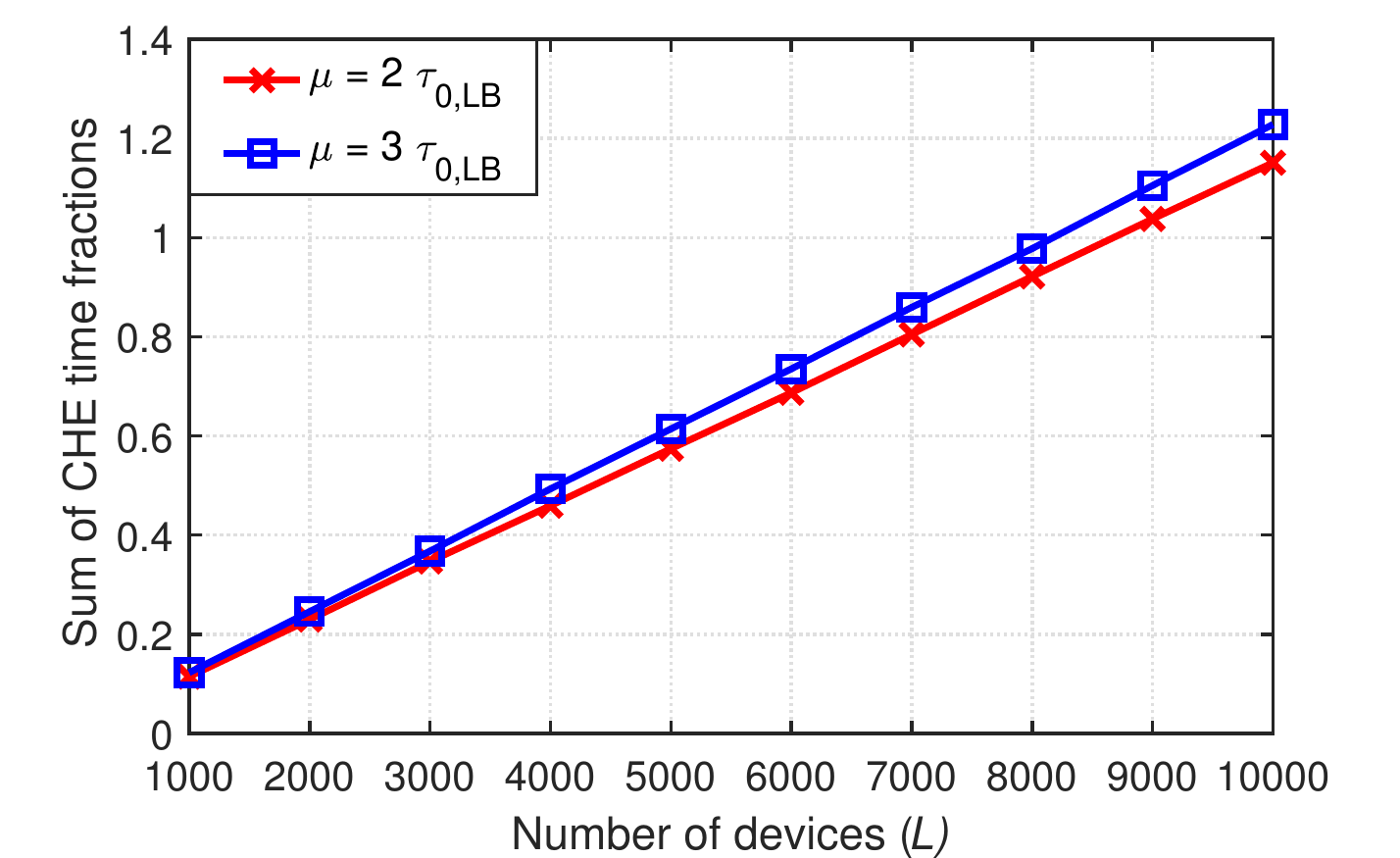}
	\caption{Sum of CHE time fractions vs. the number of devices for different values of the mean of the distribution of  the level-0 devices' time fractions.
	}\vspace{-0.7 cm}\label{sumtau}
\end{figure}


%

\begin{figure}[t]
\minipage{0.45\textwidth}
  \includegraphics[width=8 cm,height=5cm,angle=0]{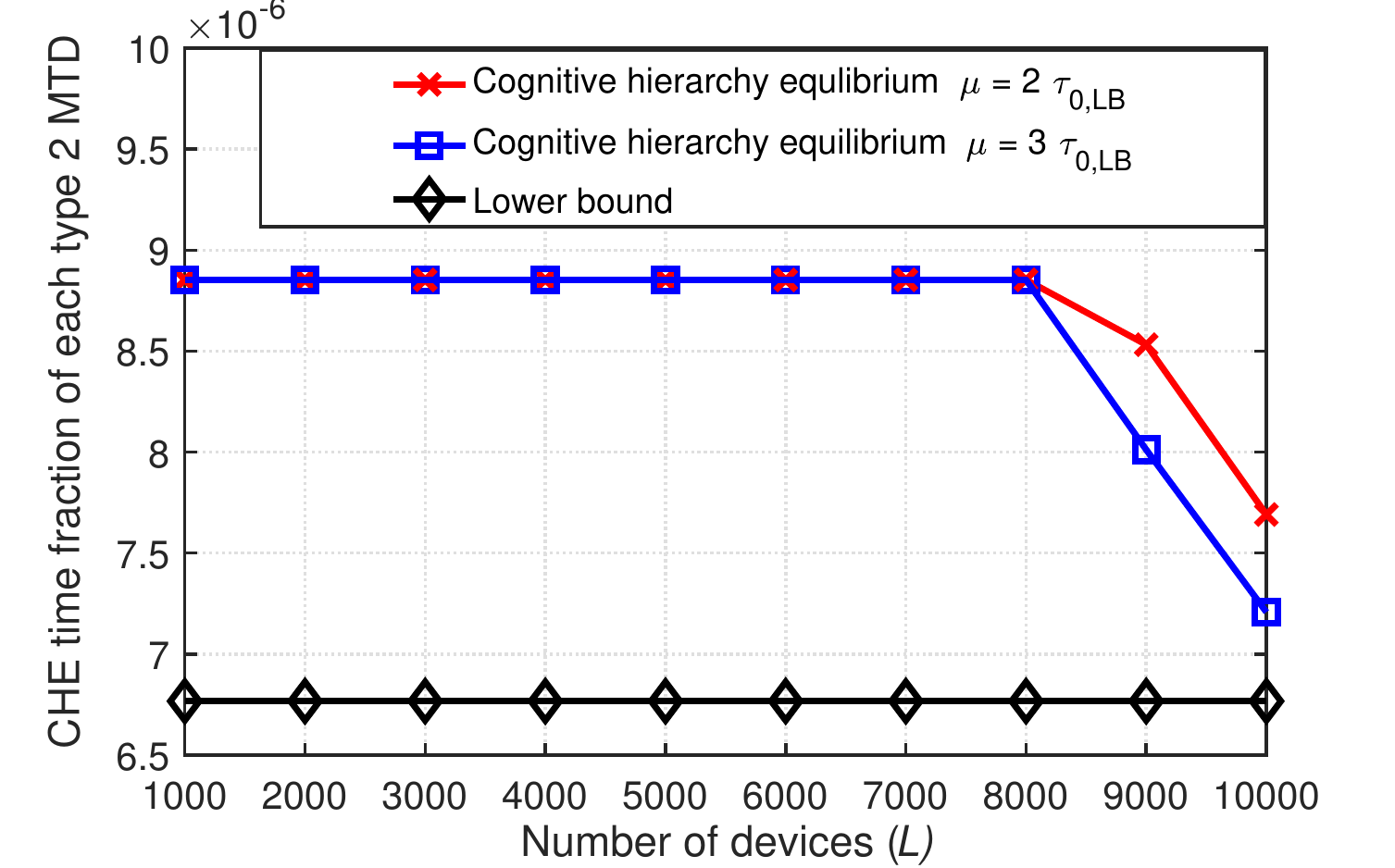}
  \caption{Normalized CHE time fraction of type 2 MTDs vs. number of devices.}\label{taum2ch}
\endminipage \hspace{1 cm}
\minipage{0.45\textwidth}
  \includegraphics[width=8 cm,height=5cm,angle=0]{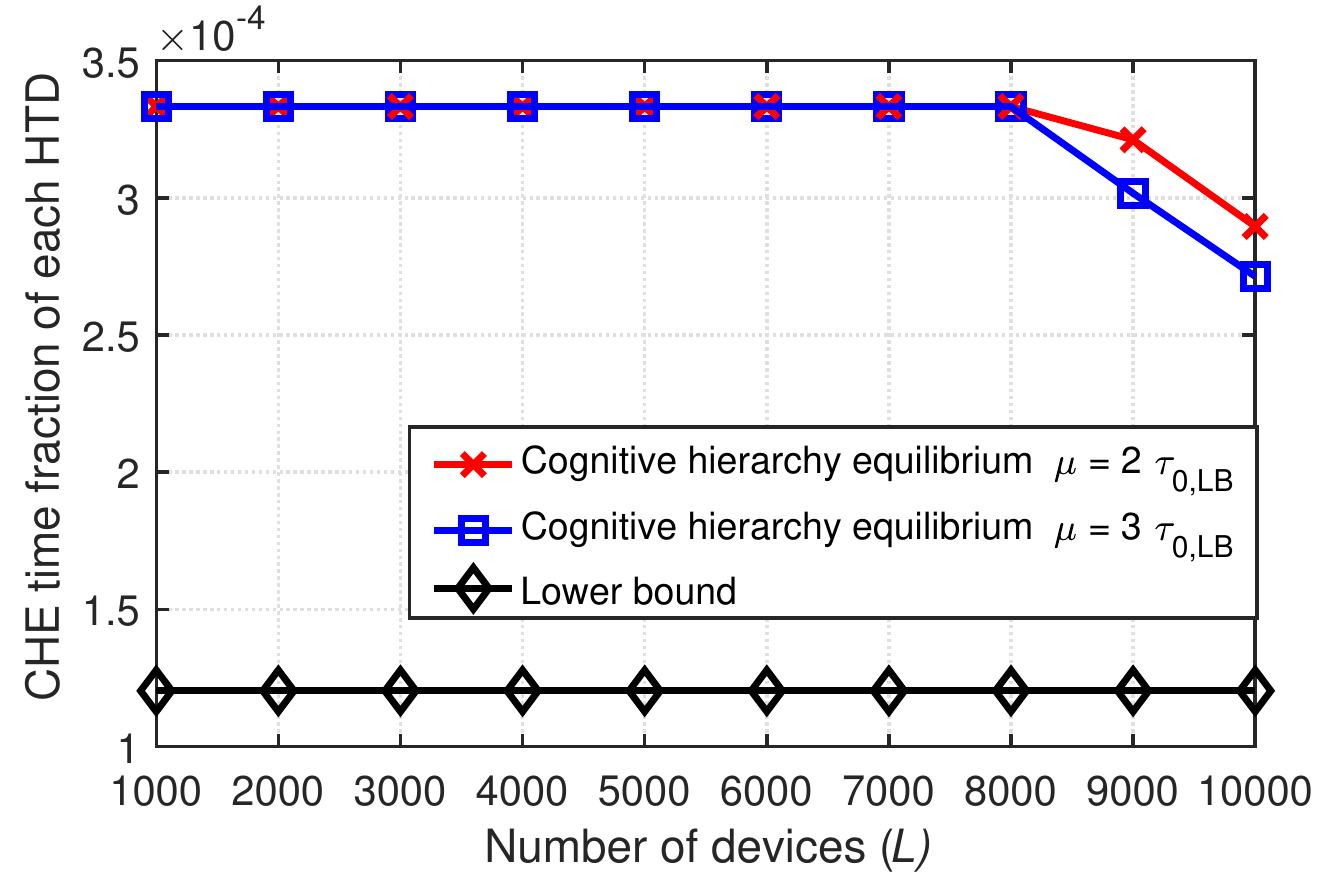}
  \caption{Normalized CHE time fraction of HTDs vs. number of devices.}\label{tauhch}
\endminipage
\vspace{-0.6 cm}
\end{figure}

Fig. \ref{taum2ch} shows the normalized CHE time fractions for type 2 MTDs, as the number of devices varies. When $\mu=2\tau_{0,LB}$ and when the number of devices is less than or equal to $8,000$,  each type 2 MTD transmits with its optimal time fraction which is $0.8853 \times 10^{-5}$. In this case, the sum of CHE time fractions of all devices is less than one as shown in Fig. \ref{sumtau}. However, since the the sum of CHE time fractions is greater than one for higher number of devices, the normalized CHE time fraction of each type $2$ MTD decreases until it reaches $0.7691 \times 10^{-5}$ when the number of devices is $10,000$.
Fig. \ref{taum2ch} also shows that when the number of devices is greater than $8,000$, the time fraction allocated for type $2$ MTDs  decreases with $\mu$ due to the increase in the sum of CHE time fractions. For a small network size, the CHE time fraction is the same for both values of $\mu$ since the sum of CHE time fractions is less than one. 

Fig. \ref{tauhch} shows the normalized CHE time fractions for HTDs as the number of devices varies.  When the number of devices is less than or equal to $8,000$ and $\mu=2\tau_{0,LB}$,  each HTD transmits with its highest possible time fraction which is $0.3333 \times 10^{-3}$. For a higher number of devices, the sum of CHE time fractions increases beyond one, and the normalized CHE time fraction for each HTD decreases until it reaches $0.2714 \times 10^{-3}$. Also, the normalized CHE time fraction of each HTD decreases with the mean $\mu$ when the number of devices is greater than $8,000$, since the sum of CHE time fractions increases with $\mu$ and is greater than one.
\vspace{-0.4 cm}
\subsection{CHE Performance Evaluation}

The performance of the CHE solution is next assessed in terms of the average total utilities of MTDs and HTDs, and the average percentage of devices with satisified QoS. 
Also, we assess the average performance of the GNE solution. The state-of-the-art GNE solution is used as a game-theoretic baseline since the GNE assumes that the players are fully rational, while the CHE solution assumes that players belong to discrete levels of bounded rationality. Further, for the considered simulation values, the GNE is unique, and, thus, performance of the GNE solution serves as a bound on the performance of the CHE solution.
This is achieved by generating 1,000 random initial vectors of the devices' time fractions. Then, for each generated random vector,  the total utility and the percentage of devices with satisfied QoS of the resulting GNE are computed. Then, the average total utility and the average percentage of devices are computed. We also consider, a non game-theoretic baseline, the ``equal time policy'' that splits the time duration $T$ equally among the IoT devices.

\begin{figure}[t]
	\centering
	\includegraphics[width=7.5 cm,height=4.5 cm,angle=0]{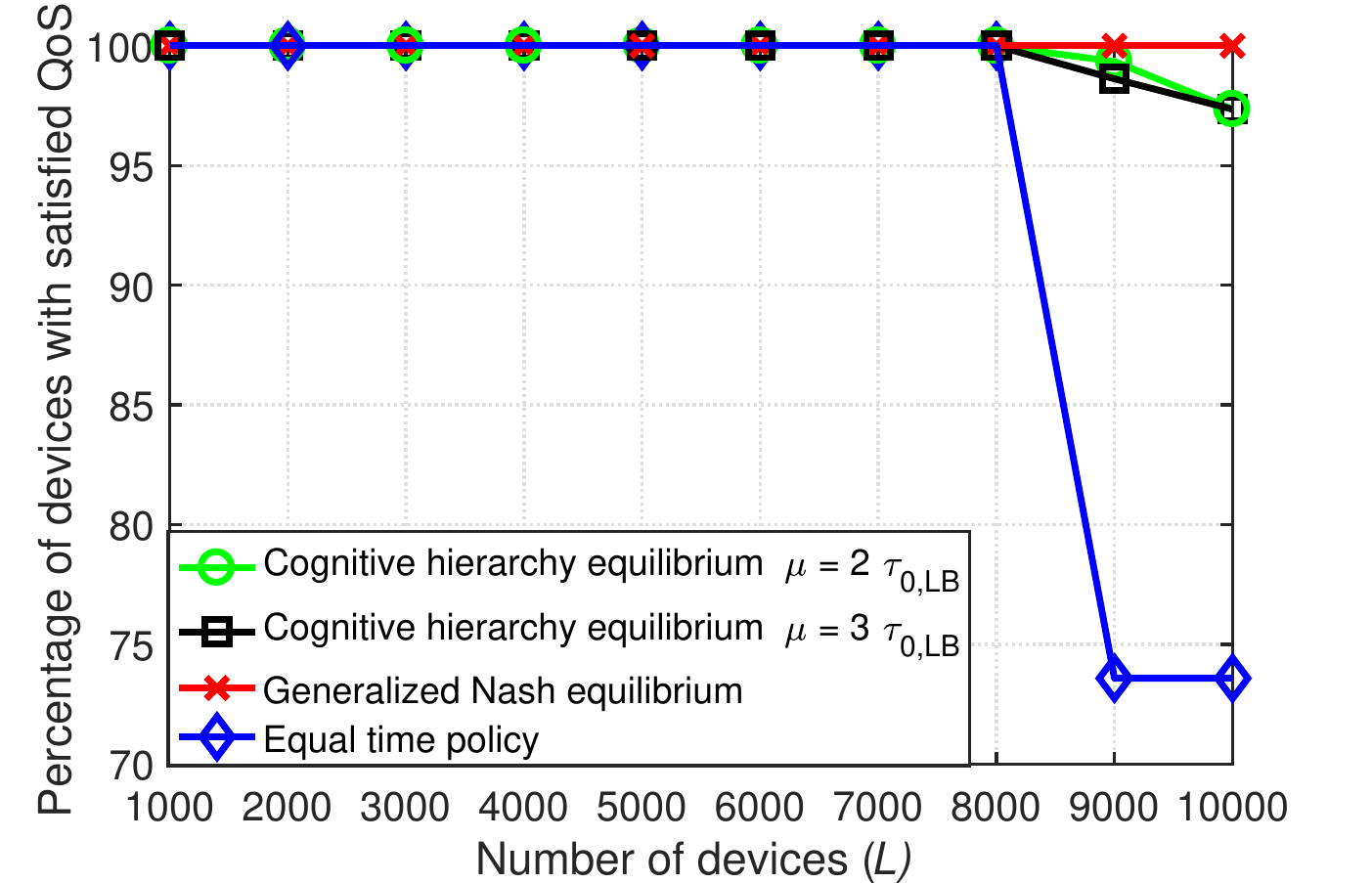}
	\caption{Percentage of devices with satisfied QoS vs. number of devices for different values of the mean of the distribution of  the level-0 devices' time fractions.
	}\vspace{-0.8 cm}\label{pernash}
\end{figure}

Fig. \ref{pernash} shows the percentage of devices with satisfied QoS constraints resulting from the CHE, the average GNE, and the equal time policy as a function of the number of devices. For the GNE, the percentage is maintained at $100\%$ for all network sizes. For the  CHE and when $\mu=2\tau_{0,LB}$, Fig. \ref{pernash} shows that the percentage of devices with satisfied QoS constraints is $100\%$ when the number of devices less than or equal to $8,000$. When the number of devices is greater than $8,000$, the percentage of of devices with satisfied QoS constraints decreases slightly until it reaches $97.35\%$. This decrease is mainly due to the fact that the normalized time fractions of some of the CH level-0 devices drops below the lower bound, whereas the time fraction of type 2 MTDs and HTDs are above the lower bound as shown in Figs. 4 and 5 in the revised manuscript. When $\mu=3\tau_{0,LB}$, the CHE maintains the same percentage of devices with satisfied QoS constraints as the case when  $\mu=2\tau_{0,LB}$ for network sizes less than or equal $8,000$. Then, the percentage of devices with satisfied QoS drops to $96\%$ as the network size increases to $10,000$.
As for the equal time policy, from Fig.  \ref{pernash}, we can see that the percentage of devices with satisfied QoS constraints is $100\%$ for a network size less than $8,000$. When the number of devices is greater than $8,000$, the percentage of devices with satisfied QoS constraints decreases to $73\%$. This is because, in this case, the time fraction assigned to each device drops below the lower bound on the time fraction of each HTD. Thus, Fig.  \ref{pernash} shows that CHE can maintain stable performance, in terms of the percentage of devices with satisified QoS, as good as the GNE solution and always outperforming the equal policy solution.


\begin{figure}[t]
\minipage{0.45\textwidth}
  \includegraphics[width=8 cm,height=4.5cm,angle=0]{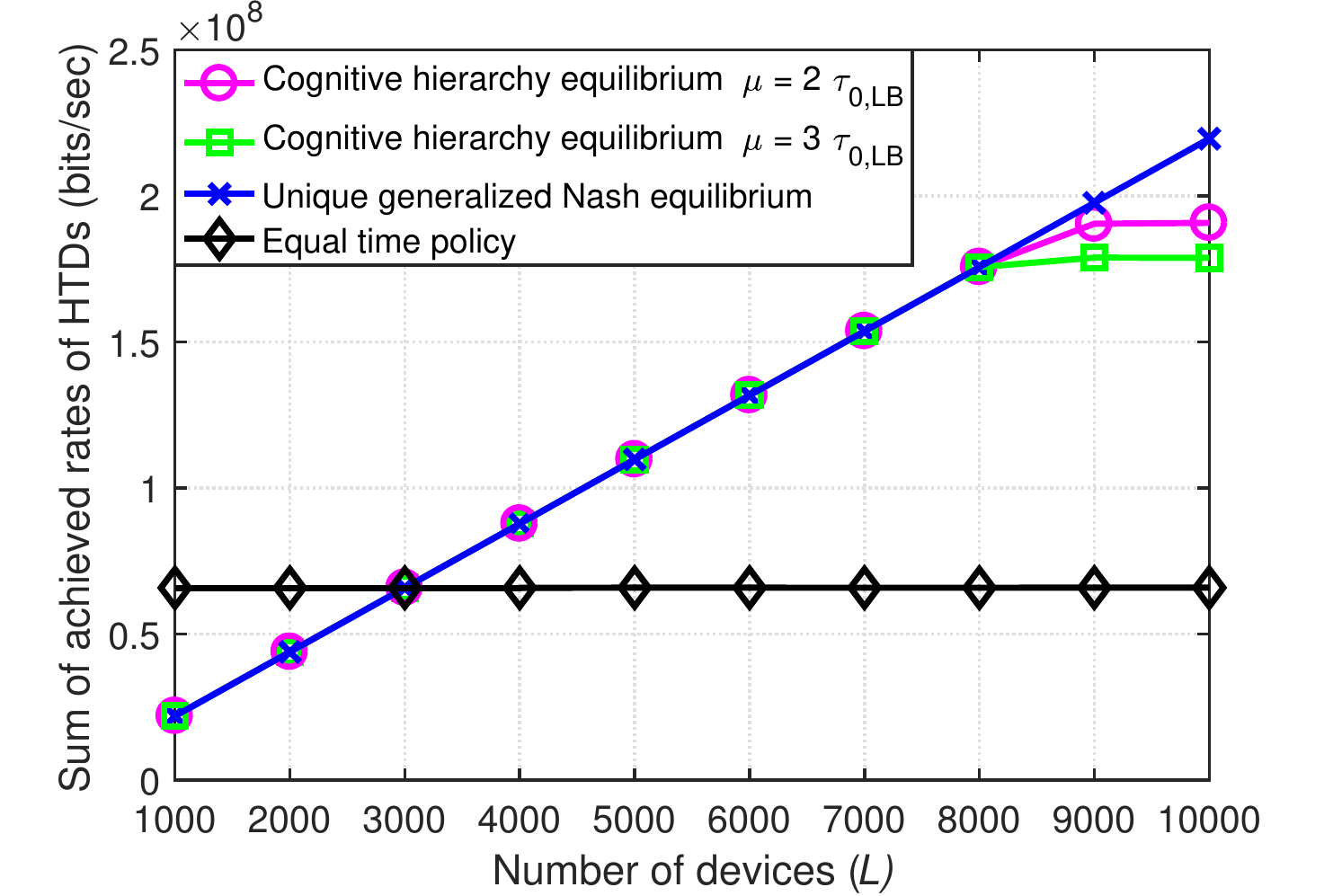}
  \caption{Total rate of HTDs vs. number of devices.}\label{avgchh2}
\endminipage \hspace{1 cm} \vspace{0.3 cm}
\minipage{0.45\textwidth}
  \includegraphics[width=8 cm,height=4.5cm,angle=0]{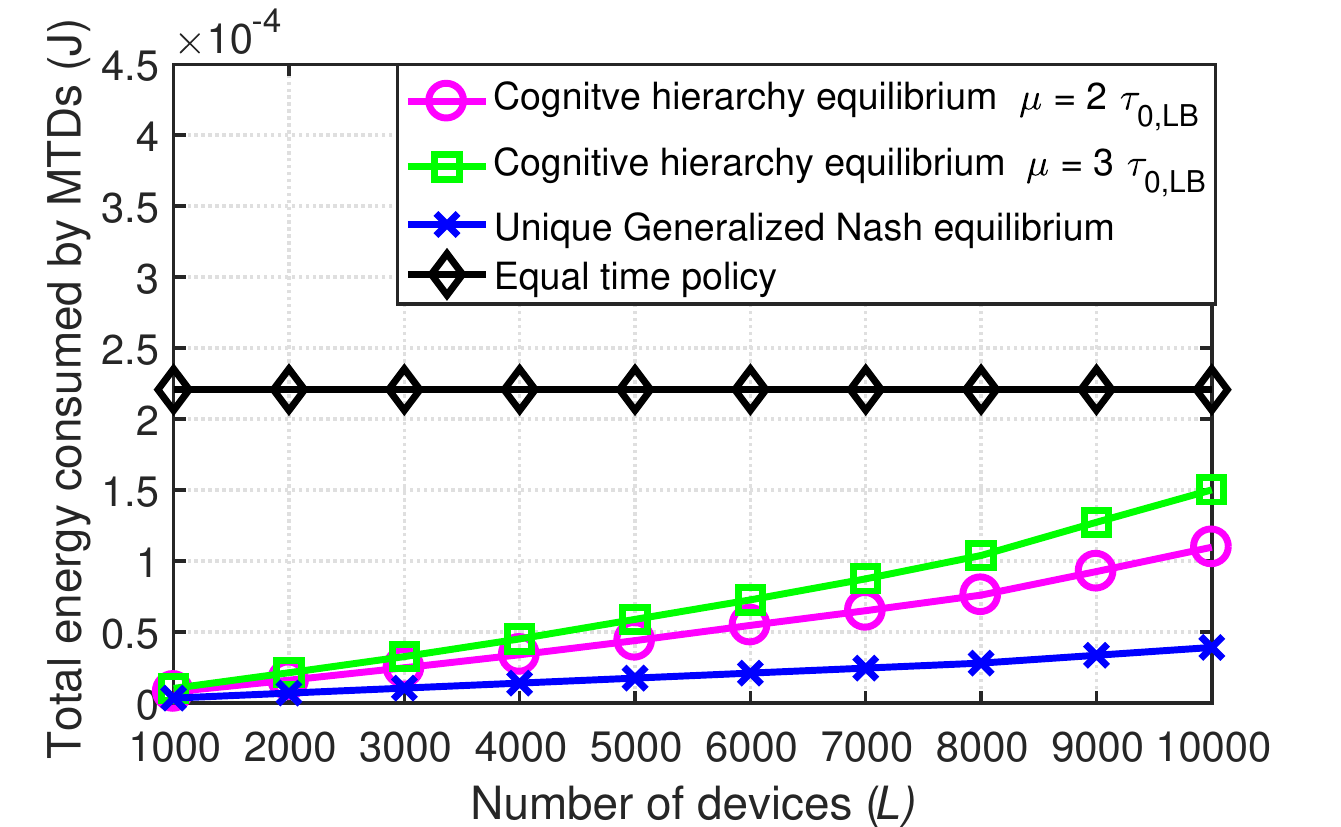}
  \caption{Total energy consumed by MTDs vs. number of devices.}\label{avgchm2}
\endminipage
\vspace{-0.8 cm}
\end{figure}

%
Next, we compute the minimum, maximum, and average total rate of HTDs and total energy of MTDs achieved by the GNE solution, the average total rate of HTDs and total energy of MTDs achieved by the CHE solution, and the total rate of HTDs and the total energy of MTDs of the equal time policy. Fig. \ref{avgchh2} shows the total achieved rate of HTDs versus the network size for the considered value of the mean $\mu$. 
For all network sizes, the GNE solution is unique and thus the minimum, maximum, and average total rates of HTDs achieved at the GNE are equal and increasing with the network size. 
For the CHE solution, when $\mu=2\tau_{0,LB}$, the total rate of the CHE solution is the same as the total rate of the GNE solution for network sizes less than $8,000$. This is because the CHE time fraction of each HTD is the optimal solution, and the sum of CHE time fractions is less than 1 as shown in Figs. \ref{sumtau} and \ref{tauhch}. For network sizes larger than $8,000$, the total rate of the CHE becomes less than total rate of the GNE solution since the CHE time fractions of all devices decreases due to normalization.
When $\mu$ is increased to $3\tau_{0,LB}$, the total rate achieved by CHE decreases for network sizes greater than $8,000$ since the sum of CHE time fractions increases and is greater than one.
For the equal time policy and for all considered network sizes, the total rate stays fixed at $65.832$ Mbits/sec. This is because the time fraction assigned to each device decreases with the network size and becomes less than the HTD optimal value for the considered network sizes. Also, the value of the total rate is fixed since we are considering HTDs of the same type. 
Thus, Fig. \ref{avgchh2} shows that the CHE solution maintains the same performance as the GNE for network sizes less than 8,000. For network sizes larger than $8,000$, the degradation of the total rate using the CHE solution, compared to the GNE solution, is only around $11\%$. Also, the CHE solution can be bring a two-fold increase in the total rate of HTDs compared to the equal time policy.

Fig. \ref{avgchm2} shows the total energy consumed as a function of the network size. Here, we note that, for the considered networks, the GNE is unique and the total energy consumed by MTDs at this GNE is increasing with the network size. 
For the CHE, when $\mu=2\tau_{0,LB}$, the average total energy consumed by MTDs is higher than the total energy of the GNE solution for all considered network sizes. This is due to the fact that, due to their limited capabilities,  level-0 devices choose their time fraction randomly while in the GNE solution level-0 devices choose the optimal time fraction. When $\mu$ is increased to $3\tau_{0,LB}$, the total energy consumed increases for all considered network sizes since level-0 devices transmit with higher time fractions. 
For the equal time policy, the total energy consumed by MTDs stays fixed at $0.2208$ mJ.
Fig. \ref{avgchm2} shows that the energy consumed by the CHE solution, when $\mu=2\tau_{0,LB}$, is reduced by around $78\%$ compared to the equal time policy.

In order to assess the efficiency of the CHE solution, we define a performance metric similar to the PoA, but tailored to the CH approach. We call this metric the \emph{price of bounded rationality (PoB)}. 
\vspace{-0.5 cm}
\subsection{Price of Bounded Rationality}

The PoB is defined as the ratio of the optimal total utility and the total utility achieved using the CHE strategy. Due to the different performance metrics between MTDs and HTDs and similar to the case of PoA, we define a PoB separate for MTDs and another PoB for HTDs. The PoB of MTDs is defined as the ratio of the total energy consumed by MTDs using the CHE solution and the minimum total energy consumed by MTDs. The PoB of HTDs, on the other hand, is defined as the ratio of the maximum total rate of HTDs and the total rate of HTDs achieved by the CHE solution.
\begin{figure}[t]
\minipage{0.45\textwidth}
  \includegraphics[width=7.5 cm,height=4.5cm,angle=0]{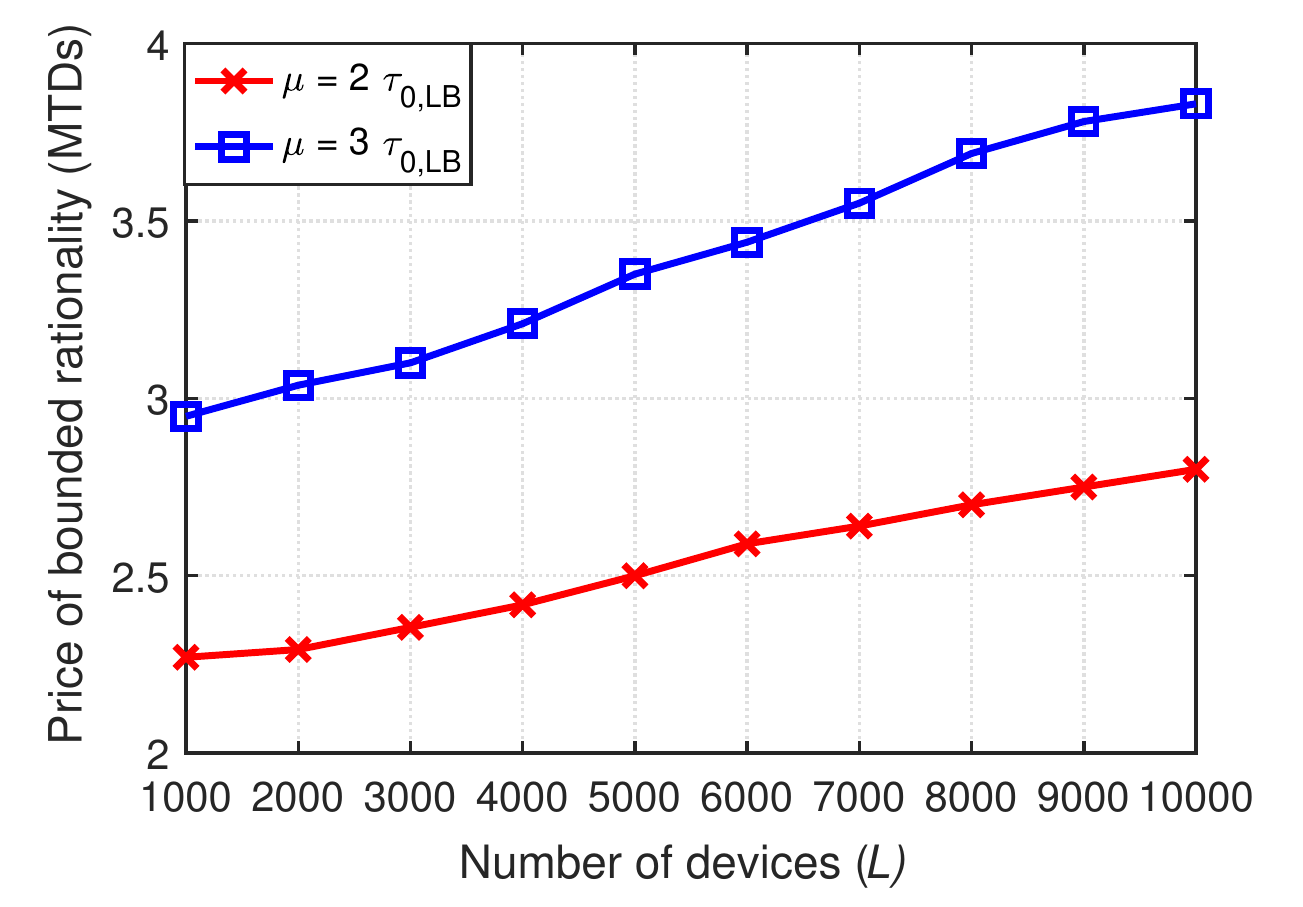}
  \caption{Price of bounded rationality of MTDs vs. number of devices.}\label{PoBm}
\endminipage \hspace{1 cm}
\minipage{0.45\textwidth}
  \includegraphics[width=7.5 cm,height=4.5cm,angle=0]{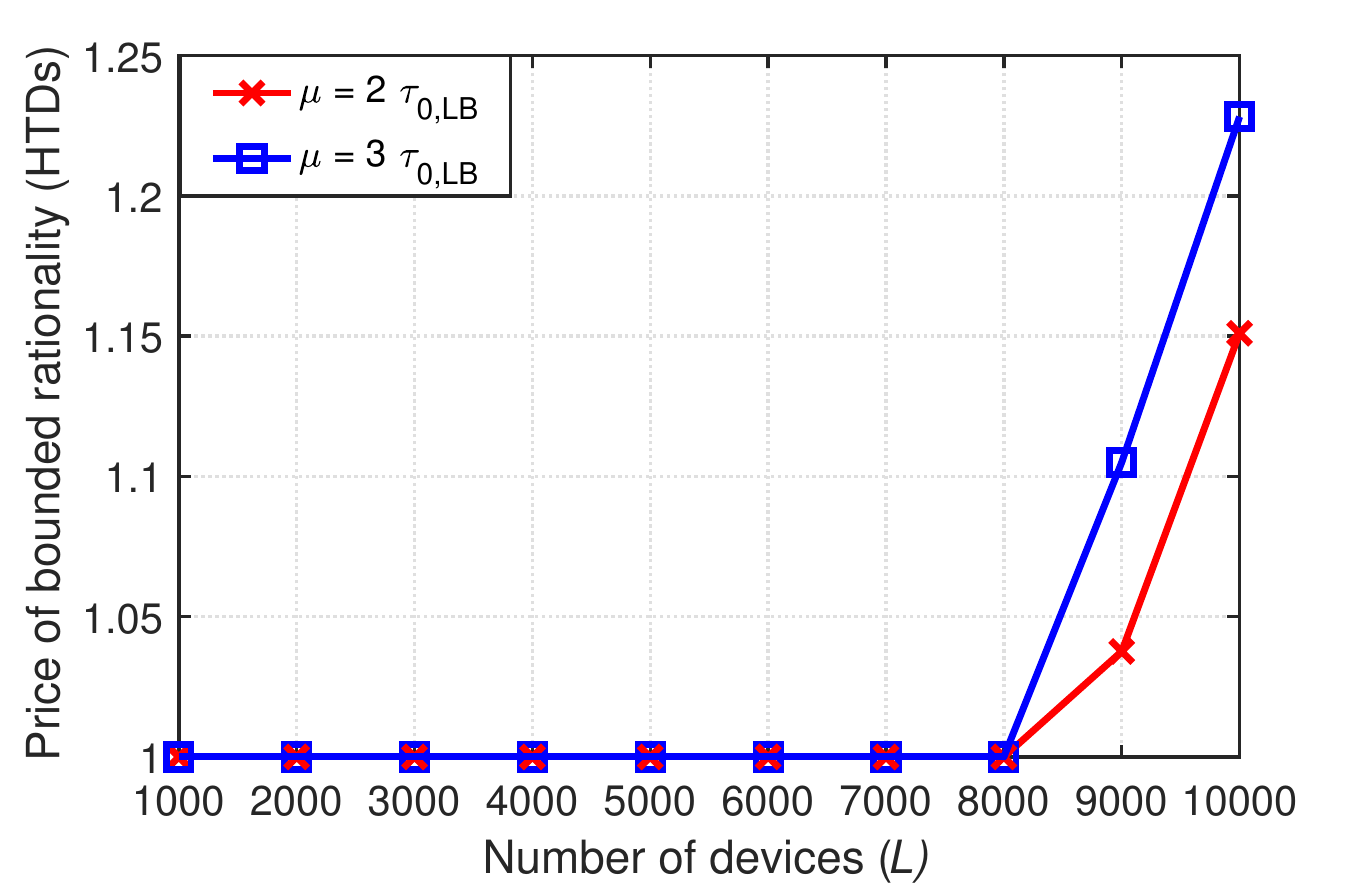}
  \caption{Price of bounded rationality of HTDs vs. number of devices.}\label{PoBh}
\endminipage
\vspace{-0.8 cm}
\end{figure}
Fig. \ref{PoBm} shows the PoB of MTDs versus the IoT network size for the considered values of $\mu$. From this figure, we can see that when $\mu=2\tau_{0,LB}$, the PoB of MTDs is around 2.27 when the network size is 1,000. Then, the PoB increases with the network size until it reaches $2.81$ when the network size is $10,000$. This is due to the fact that the number of level-0 devices, that choose their time fraction randomly, increases with the network size which results in a higher energy consumption compared to the GNE solution.
When $\mu$ increases to $3\tau_{0,LB}$,  the PoB of MTDs increases for each considered network size. This increase is mainly due  to the fact that as $\mu$ increases, the probability that a CH level-0 device transmits with higher time fractions increases, which increases the total energy consumed. Fig. \ref{PoBh} also shows that the PoB of MTDs increases at a low rate with the network size. Thus, the proposed CH approach can clearly maintain a stable performance of MTDs for larger network sizes.

Fig. \ref{PoBh} shows the PoB of HTDs versus the IoT network size for the considered values of the mean $\mu$. When the network size is less than 8,000 and when $\mu=2\tau_{0,LB}$, the PoB of HTDs is $1$. This is because the CHE time fraction of each HTD is the time fraction that maximizes its utility. As the network size increases beyond 8,000, the PoB of HTDs increases until it reaches $1.15$ when the network size is $8,000$. This increase in PoB is due to decrease in the normalized CHE time fraction of each HTD as shown in Fig. 5 in the revised manuscript. When $\mu$ increases to $3\tau_{0,LB}$, the PoB of HTDs increases for network sizes larger than $8,000$ due to the decrease in the normalized CHE time fraction according to Fig. 5 in the revised manuscript. Fig. \ref{PoBh} shows that the PoB of HTDs increases at a low rate with the network size. Thus, the proposed CH approach can maintain stable performance of HTDs for larger network sizes.

\vspace{-0.5 cm}
\section{Conclusion}
In this paper, we have considered the problem of distributed uplink time allocation in an IoT network where the IoT devices have heterogeneous quality-of-service requirements. We have formulated the problem as a noncooperative game that takes into account the heterogeneous requirements of the IoT devices. In this game, the players are the IoT devices, and their actions are to choose the time fractions necessary to ensure their quality-of-service requirements. In the proposed game, the strategy set of each device is dependent on the actions taken by the other devices. Hence, we have first characterized the set of GNEs. Moreover, we have proposed an algorithm to find the GNE of the devices and shown that the computational complexity is polynomial in the number of devices. Then, we have proposed a novel solution using CH theory to take into account the heterogeneous computational capabilities of the IoT devices. Thus, the proposed CH solution provides a more realistic solution than the GNE that assumes that all players have the same capabilities. We have characterized the cognitive hierarchy equilibrium and compared it analytically to the GNE. Extensive simulations have been conducted to thoroughly assess the various performance tradeoffs of the proposed approach. Finally, we note that, beyond the IoT application treated here, the proposed cognitive hierarchy framework can be generalized to any wireless network in which heterogeneity and bounded rationality are key features.


%
%

\vspace{-0.3 cm}

\def\baselinestretch{0.87}

\vspace{-0.4 cm}
\appendices

\section{}
\subsection{Proof of Lemma 1}
By applying the $\log$ function to the energy function $\chi_i(\gamma_i)= \frac{ P_i b_i}{ W\log(1+\gamma_i)e^{-\frac{\gamma_i \sigma^2}{\alpha_i^2P_i}}}$ we have
	$\log \chi_i(\gamma_i)=\log P_ib_i-\log (W)-\log\log(1+\gamma_i)+ \frac{\gamma_i \sigma^2}{\alpha_i^2P_i}$. The term $-\log\log(1+\gamma_i)$ is convex in $\gamma_i$ since $\log \log(1+\gamma_i)$ is concave in $\gamma_i$. Also, the term $\frac{\gamma_i \sigma^2}{\alpha_i^2P_i}$ is linear in $\gamma_i$. Hence, the function $\log \chi_i(\gamma_i)$ is convex in $\gamma_i$. It follows that the energy function $\chi_i(\gamma_i)$ is convex in $\gamma_i$ as it logconvex in $\gamma_i$. Then, the utility of MTD $i$ is concave since $U_i(\gamma_i)=-\chi_i(\gamma_i)$.

\vspace{-0.2 cm}
\subsection{Proof of Proposition 1}
For each HTD $i$,  the utility function is linear in its action $a_i=\tau_i$ hence it is also concave. For each MTD $i$, it is shown in Lemma 1 that the utility is a concave in its action $a_i=\gamma_i$. Also, for each device $i$ the strategy space $\mathcal{S}_i(\boldsymbol{a}^*_{-i})$ is nonempty, closed and convex. The result follows directly from \cite[Theorem 4.1]{gnep}.
\vspace{-0.2 cm}
\subsection{Proof of Proposition 2}
The utility of HTD $i$ is an increasing linear function of $\tau_i$. Thus, the optimal value of $\tau_i$ is its upper bound. From the constraints of (\ref{HTDopt}), we get (\ref{brh1}). 
\vspace{-0.2 cm}
\subsection{Proof of Proposition 3}
As shown in Lemmas 1 and 2, the utility of MTD $i$ is concave and attains its maximum at $\gamma'_i$. Hence, given any strategy vector $\boldsymbol{a}_{-i}$, MTD $i$ chooses $a_i=\gamma'_i$ if $\gamma'_i \in \mathcal{S}_i(\boldsymbol{a}_{-i})$. Otherwise, MTD $i$ chooses the upper bound of its strategy set $a_i=-\frac{\alpha_i^2 P_i\log(1-\sqrt[t_i]{\epsilon})}{\sigma^2}$ if the upper bound is less than $\gamma'_i$ since the utility is increasing over the strategy set in this case. The last case is when the lower bound is greater than $\gamma'_i$, MTD $i$ chooses the lower bound since the utility function is decreasing over the strategy set.
\vspace{-0.2 cm}
\subsection{Proof of Proposition 4}
Since MTD $i$ utility is concave in $\gamma_i$, the result follows using a similar argument as the proof of Proposition \ref{propGNEMTD}. For HTD $i$, the utility function is increasing in its strategy $a_i$. Thus, the optimal value is the upper bound of the strategy space of HTD $i$.
\vspace{-0.4 cm}
\section{}
\subsection{Proof of Theorem \ref{GNEthm}}
For the first case i.e. when $\sum_{i \in \mathcal{M}'}\frac{b_i}{T\cdot W \log(1+\gamma'_i)}+\sum_{i \in \mathcal{H}}\frac{E_i}{ T\cdot WP_i}+ \sum_{i \in \mathcal{M}-\mathcal{M}'}\frac{b_i}{T\cdot W \log(1+\gamma_{i,\textrm{UB}})} \leq 1$,  it is clear according to the best response equations in propositions (\ref{propGNEHTD}) and (\ref{propGNEMTD}) that for the action profile $a^{**}_i=\gamma'_i$ $\forall i \in \mathcal{M}'$ and $a^{**}_i=-\frac{\alpha_i^2 P_i\log(1-\sqrt[t_i]{\epsilon})}{\sigma^2}$ $\forall i \in \mathcal{M}-\mathcal{M}'$, and $a^{**}_i=\frac{E}{ T\cdot WP_i}$ for each HTD $i$, no device has the incentive to change its strategy since $a^{**}_i$ is the optimal solution of the utility of device $i$. Next, we show that there exists no other action profile that constitutes a GNE. For any action profile $\boldsymbol{a}'$ other than $\boldsymbol{a}^{**}$, a device $i \in \mathcal{M}'$ will always change its strategy to the optimal strategy $\gamma'_i$ if $a'_i <\gamma'_i$ since  $\gamma'_i$ is the optimal solution of its utility and it yields a lower time fraction. Also, any device $i \in \mathcal{M}- \mathcal{M}'$ has an incentive to change its strategy to the optimal strategy $\gamma_{i,\textrm{UB}}$ if $a'_i<\gamma_{i,\textrm{UB}}$ resulting in a lower time fraction. Similarly, any HTD $i$ will reduce its strategy to $\frac{E_i}{ T\cdot WP_i}$ if $a'_i>\frac{E_i}{ T\cdot WP_i}$ since $\frac{E_i}{ T\cdot WP_i}$ is the upper bound. Let $\mathcal{D}$ be the set of all such devices. For the newly formed action profile $\boldsymbol{a}''$ s.t. $a''_i=a^{**}_i$ if $i \in \mathcal{D}$ and $a''_i=a'_i$ if $i \in \mathcal{L} - \mathcal{D}$ , $\sum_{i \in  \mathcal{D} \cap \mathcal{M}}\frac{b_i}{T\cdot W \log(1+a^{**}_i)}+\sum_{i \in \mathcal{D} \cap \mathcal{H}}\frac{E_i}{ T\cdot WP_i}+ \sum_{i \in (\mathcal{L} - \mathcal{D}) \cap \mathcal{M}}\frac{b_i}{T\cdot W \log(1+a'_i)}+ \sum_{i \in  (\mathcal{L} -\mathcal{D}) \cap \mathcal{H}}a'_i<\sum_{i \in \mathcal{M}'}\frac{b_i}{T\cdot W \log(1+\gamma'_i)}+\sum_{i \in \mathcal{H}}\frac{E_i}{ T\cdot WP_i}+ \sum_{i \in \mathcal{M}-\mathcal{M}'}\frac{b_i}{T\cdot W \log(1+\gamma_{i,\textrm{UB}})} \leq 1$ since $U(a'_i)<U(a^{**}_i)$ for $a'_i > a^{**}_i$ for each MTD in $\mathcal{L}-\mathcal{D}$ and $a'_i < a^{**}_i$ for each HTD in $\mathcal{L}-\mathcal{D}$, and, hence, each device in $\mathcal{L}-\mathcal{D}$ has an incentive to change its action to $a^{**}_i$.
	
	For the second case i.e. $\sum_{i \in \mathcal{M}'}\frac{b_i}{T\cdot W \log(1+\gamma'_i)}+\sum_{i \in \mathcal{H}}\frac{E_i}{ T\cdot WP_i}+ \sum_{i \in \mathcal{M}-\mathcal{M}'}\frac{b_i}{T\cdot W \log(1+\gamma_{i,\textrm{UB}})} > 1$ 
	for any action profile $\boldsymbol{a}' \notin \mathcal{N}$. Similar to the first case, any device $i \in \mathcal{D}$ ($\mathcal{D}$ being same set previously defined) has an incentive to change its strategy to $\gamma'_i$ if $i \in \mathcal{D}\cap \mathcal{M}'$ and $\gamma_{i,\textrm{UB}}$ if $i \in \mathcal{D}\cap (\mathcal{M}-\mathcal{M}')$ and to $\frac{E_i}{ T\cdot WP_i}$ if $i \in \mathcal{D} \cap \mathcal{H}$. For the newly formed action profile $\boldsymbol{a}''$ s.t. $a''_i=a^{**}_i$ if $i \in \mathcal{D}$ and $a''_i=a'_i$ if $i \in \mathcal{L} - \mathcal{D}$. It can be easily verified that each action $a^* \in \mathcal{N}$ is a GNE. No device in $\mathcal{A}_j$ has any incentive to change its strategy since $a^*_i$ is the optimal value of its utility and all other devices in $\mathcal{L}-\mathcal{A}_j$ cannot improve their utility since the sum of time fractions corresponding to $a^*$ is one and an improvement in their utility will need a higher allocated time fraction.
\vspace{-0.3 cm}
\subsection{Proof of Theorem 2}
First, we show that Algorithm 1 converges in at most three iterations. Without loss of generality, we assume that the initial sum of time fractions allocated to the devices is one. In the first iteration of the algorithm, all MTDs that can improve their utility by reducing their current time fraction will change their strategy. If there are no such MTDs that can improve their utility in the first round, the algorithm stops since no device can change its strategy. The worst-case scenario occurs when during the first round the last MTD that changes its strategy by reducing its time fraction is the last in the order. This is because by the end of the first round, the sum of time fractions of all devices is less than one. Hence, in the second round, it is only possible for the devices to improve their utilities by choosing a higher time fraction. The algorithm terminates in the third round because no other device can further improve its utility. 
\vspace{-0.2 cm}
\subsection{Proof of Corollary 1}
Based on Theorem \ref{complexityGNE}, the complexity of the computations done at each device, to find the GNE, is $O(L)$ since Algorithm 1 converges in at most three iterations. Each device solves at most two best response optimization problems to converge to its GNE strategy. The complexity of finding the best response based on Propositions \ref{propGNEHTD} and \ref{propGNEMTD} is $O(L)$.
\vspace{-0.3 cm}
\subsection{Proof of Theorem \ref{propCHcomplexity}}
To find the CHE, a device at level $k$ needs to evaluate CHE strategies of all devices assuming that they are level $1$ to $k$ based on its beliefs. Based on Propositions 4 and 5, finding the strategy of each device takes $O(1)$ times. Hence, to find the strategy of all devices of each level takes $O(C)$ times. Thus, to find the strategy of devices for all levels up to $k$ takes $O((k+1)\cdot C)$.
\vspace{-0.3 cm}
\subsection{Proof of Theorem \ref{theoremcompare}}
For the first case, we start by investigating the strategy profile $\boldsymbol{d}$. Since level-0 MTDs choose their time fractions randomly, the CHE strategy of level-0 MTD $i$ could possibly be $\gamma_{\nu_i}<\min\{\gamma'_i,\gamma_{i,UB}\}$. In this case, the normalized time fraction of level-0 device will be greater than the time fraction that maximizes its utility. Hence, level-0 MTD $i$ has the incentive to change its strategy to $\min\{\gamma'_i,\gamma_{i,UB}\}$, and, hence, the CHE solution will not be a GNE. In the second case when $\gamma_{\nu_i}>\min\{\gamma'_i,\gamma_{i,UB}\}$, the normalized time fraction of level-0 MTD $i$ will be less than the optimal time fraction. In this case, level-0 MTD $i$ cannot increase its time fraction since the sum of the CHE time fractions is one.
Now, for any MTD $i$ belonging to a level higher than 0, if $\gamma'_i\geq \gamma_{i,\textrm{UB}}$, the CHE strategy of MTD $i$ will be $\gamma_{i,\textrm{UB}}$ and  $\gamma_{\nu_i}$ will be greater than or equal to $\gamma_{i,\textrm{UB}}$. In this case, the normalized time fraction $\nu_i$ will be lower than the lower bound of the time fraction (which is the value of the time fraction that corresponds to $\gamma_{i,\textrm{UB}}$). Hence, MTD $i$ has no incentive to further reduce its allocated time fraction.  For the case in which $e^{\frac{b_k}{\frac{T\cdot W}{g_k(k)} \Big(1-\sum_{h=0}^{k-1}g_k(h)\sum_{q \in \mathcal{C}_k}N_q\frac{b_h}{T\cdot W\log(1+\gamma^*_q(h))}\Big)}}-1\leq \gamma'_i \leq \gamma_{i,\textrm{UB}}$, the CHE strategy of MTD $i$ is $\gamma'_i$ and hence $\gamma_{\nu_i}$ will be greater than or equal to $\gamma'_i$. Hence, MTD $i$ has no incentive  to increase its $\gamma_i$ since this will result in decreasing its utility. Also, MTD $i$ can not increase its time fraction by reducing its $\gamma_i$ since the sum of the normalized time fractions is one. When $e^{\frac{b_k}{\frac{T\cdot W}{g_k(k) \cdot L} \Big(1-\sum_{h=0}^{k-1}g_k(h)\sum_{q \in \mathcal{C}_k}N_q\frac{b_h}{T\cdot W\log(1+\gamma^*_q(h))}\Big)}}-1>\gamma'_i$, $\gamma_{\nu_i}$ will be greater than $\gamma'_i$ using an argument similar to the previous case, hence, MTD $i$ has no incentive to change its strategy $d_i$. Any HTD $i$ cannot increase its allocated time fraction $\nu_i$ since the sum of the normalized time fractions is one.

For the second case, if  $\sum_{i \in \mathcal{M}'}\frac{b_i}{T\cdot W \log(1+\gamma'_i)}+\sum_{i \in \mathcal{H}}\frac{E}{ T\cdot WP_i}+ \sum_{i \in \mathcal{M}-\mathcal{M}'}\frac{b_i}{T\cdot W \log(1+\gamma_{i,\textrm{UB}})} > 1$, we know from Theorem \ref{GNEthm} that the GNE is not unique and that the sum of time fractions is one. Hence, the CHE strategy is not a GNE. Thus, for the CHE strategy $\boldsymbol{d}$, there exists a subset of devices that can improve their utility until the sum of time fractions is one.
Otherwise when
$\sum_{i \in \mathcal{M}'}\frac{b_i}{T\cdot W \log(1+\gamma'_i)}+\sum_{i \in \mathcal{H}}\frac{E}{ T\cdot WP_i}+ \sum_{i \in \mathcal{M}-\mathcal{M}'}\frac{b_i}{T\cdot W \log(1+\gamma_{i,\textrm{UB}})} \leq 1$ , we know from Theorem \ref{GNEthm} that the GNE is unique and hence based on Theorem \ref{GNEthm} the CHE strategy is a GNE according to Theorem \ref{GNEthm} if the CHE strategy is for any MTD $i$ $d_i(k_i)=\gamma'_{i}$ if $i \in \mathcal{M}'$ and $d_i(k_i)=\gamma_{i,\textrm{UB}}$ if $i \in \mathcal{M}-\mathcal{M'}$. Otherwise, any other CHE strategy profile $\boldsymbol{d}$ will have a lower performance in terms of the total utility than that of the GNE since every action $a^{**}_i$ in the GNE strategy profile $\boldsymbol{a}^{**}$ maximizes device $i$ utility. Here, we also recall that the utility function of each device is independent of the actions of the remaining devices. Thus, the performance of the CHE, due to bounded rationality, is upper bounded by the GNE.

\end{document}